%% file: rrp-nocomments.tex
\documentclass[smallcondensed,nospthms]{svjour3}
\input{rrp-macro-file.tex}
\RequirePackage{fix-cm}
\usepackage{mathptmx}      %
\smartqed  %

\title{Robust computation of linear models by convex relaxation\thanks{Communicated by Emmanuel Cand\`es}}
\titlerunning{Robust computation of linear models}

\author{Gilad Lerman \and Michael B. McCoy  \and \\ Joel A. Tropp \and Teng
  Zhang}
\authorrunning{G. Lerman et al.}

\institute{G. Lerman \at Department of Mathematics, University of
  Minnesota, 127 Vincent Hall, 206 Church St SE, Minneapolis, MN 55455
  USA, \email{lerman@umn.edu}
  \and M. B. McCoy (\emph{corresponding author}) \at Cofacet, Inc. 693 E Del Mar Blvd., Pasadena, CA 91101 \email{mccoy@cofacet.com}
  \and J. A. Tropp\at Department of Computing
  and Mathematical Sciences, California Institute of Technology, MC
  305-16 1200 E. California Blvd.  Pasadena, CA 91125, \email{jtropp@cms.caltech.edu}
  \and T. Zhang \at University of Minnesota, Institute for Mathematics and its
  Applications, 207 Church Street SE, Minneapolis, MN 55455, \email{zhang620@umn.edu}
}

\date{18 February 2012.  Revised 31 May 2013.}

\begin{document}
\maketitle
\begin{abstract}
Consider a dataset of vector-valued observations that consists
of noisy {inliers}, which are explained well
by a low-dimensional subspace, along with some number of {outliers}.
This work describes a convex optimization problem, called \rrp,
that can reliably fit a low-dimensional model to this type of data.
This approach parameterizes linear subspaces using orthogonal projectors,
and it uses a relaxation of the set of orthogonal projectors to
reach the convex formulation.
The paper provides an efficient algorithm for solving the \rrp\
problem, and it documents numerical experiments which confirm that
\rrp\ can dependably find linear structure in synthetic and natural
data.  In addition, when the inliers lie near a low-dimensional
subspace, there is a rigorous theory that describes when \rrp\
can approximate this subspace.

\end{abstract}

\section{Introduction}
\label{sec:introduction}

Low-dimensional linear models have applications in a huge array of data analysis problems.
Let us highlight some examples from computer vision, machine learning,
and bioinformatics.

\begin{description} \setlength{\itemsep}{5pt}
\item	[\textbf{Illumination models}]  Images of a face---or any Lambertian object---viewed under different illumination conditions lie near a nine-dimensional subspace~\cite{EHY95:5pm2,HYL+03:Clustering-Appearances,BJ03:Lambertian-Reflectance}.

\item	[\textbf{Structure from motion}]  Feature points on a moving rigid body lie on an affine space of dimension three, assuming the affine camera model~\cite{CK98:Multibody-Factorization}.  More generally, estimating structure from motion involves estimating low-rank matrices~\cite[Sec.~5.2]{EvdH10:Efficient-Computation}.

\item	[\textbf{Latent semantic indexing}]  We can describe a large corpus of documents that concern a small number of topics using a low-dimensional linear model~\cite{DDL+88:Improving-Information}.

\item	[\textbf{Population stratification}]
Low-dimensional models of single nucleotide polymorphism (SNP) data have been used to show that the genotype of an individual is correlated with her geographical ancestry~\cite{NJB+08:Genes-Mirror}.
More generally, linear models are used to assess differences
in allele frequencies among populations~\cite{PPP+06:Principal-Components}.
\end{description}

\noindent
In most of these applications, the datasets are noisy, and they contain a substantial number of outliers.  Principal component analysis, the standard method for finding a low-dimensional linear model, is sensitive to these non-idealities.  As a consequence, good robust modeling techniques would be welcome in a range of scientific and engineering disciplines.

In recent years, researchers have started to use convex optimization to develop alternatives to principal component analysis that have more favorable robustness properties.  For the most part, these formulations attempt to find a low-rank matrix that approximates the data well.  They typically use the Schatten 1-norm as a convex proxy for the rank.  See Section~\ref{sec:previous} for a more complete discussion.
Although these ideas are compelling, it remains valuable to
explore other methods because of the importance of linear modeling.

This paper describes a new technique for fitting a low-dimensional linear model to data.
Our formulation is based on convex optimization, but it has a different flavor from the earlier techniques.  We use a new set of ideas to develop a rigorous analysis of
the performance of our method.  This theory demonstrates that the approach is robust
against noise in the inliers, and it can cope with a large number of adversarial outliers.
We describe an efficient numerical algorithm that is guaranteed to solve the optimization problem after a modest number of spectral calculations.  We also include some experiments with synthetic and natural data to verify that our technique reliably seeks out linear structure.

\subsection{Notation and Preliminaries}

In this paper, we work with real-valued data.  We write $\norm{\cdot}$ for the $\ell_2$ norm on vectors and the spectral norm on matrices; $\fnorm{\cdot}$ represents the Frobenius norm; $\pnorm{S_1}{\cdot}$ refers to the Schatten 1-norm.   Angle brackets $\ip{\cdot}{\cdot}$ denote the standard inner product on vectors and matrices, and $\trace$ refers to the trace.  The curly inequality $\psdle$ denotes the semidefinite order: For symmetric matrices $\mtx{A}$ and $\mtx{B}$, we write $\mtx{A} \psdle \mtx{B}$ if and only if $\mtx{B} - \mtx{A}$ is positive semidefinite.

An \term{orthoprojector} is a symmetric matrix $\Proj$ that satisfies $\Proj^2 = \Proj$.
Each subspace $L$ in $\R^D$ is the range of a unique $D \times D$ orthoprojector $\Proj_L$.  The trace of an orthoprojector equals the dimension of its range: $\trace(\Proj_L) = \dim(L)$.  For each point $\vct{x} \in \R^D$, the image $\Proj_L \vct{x}$ is the best $\ell_2$ approximation of $\vct{x}$ in the subspace $L$.  The orthogonal complement of a subspace $L$ is expressed as $L^{\perp}$.

For a real number $a$, the notation $\lfloor a \rfloor$ refers to the greatest integer that does not exceed $a$, and $\lceil a \rceil$ refers to the smallest integer that is at least as large as $a$.  These operations are usually referred to as \term{floor} and \term{ceiling}, respectively.  We also define the function $[a]_+ := \max\{a, 0\}$, which returns the positive part of a real number.

Finally, we introduce the \term{spherization transform} for vectors:
\begin{equation} \label{eqn:tilde}
\tvct{x} := \begin{cases}
\vct{x} / \norm{\vct{x}}, & \vct{x} \neq \vct{0} \\
\vct{0}, & \text{otherwise}.
\end{cases}
\end{equation}
We extend the spherization transform to matrices by applying it separately to each column.

\subsection{Linear Modeling by Principal Component Analysis}

To motivate our approach to linear modeling, we summarize a classical line of research in statistics that begins with principal component analysis.

Let $\Xall$ be a dataset\footnote{A dataset is simply a finite multiset, that is, a finite set with repeated elements allowed.} consisting of $N$ points in $\R^D$.
Suppose we wish to determine a $d$-dimensional subspace that best explains
the data.
For each point, we can measure the residual error in the approximation
by computing the orthogonal distance from the point to the subspace.
The classical method for fitting a subspace asks us to
minimize the sum of the \emph{squared} residuals:
\begin{align} \label{eqn:pca}
\minimize \sum_{\vct{x} \in \Xall} \enormsq{ \vct{x} - \Proj \vct{x} }
\subjto &\text{$\Proj$ is an orthoprojector}
\quad\textsf{and}\quad \notag \\[-10pt]
&\trace \Proj = d.
\end{align}
(Here and elsewhere, sums indexed by a dataset repeat each point as many times
as it appears in the dataset.)
The approach~\eqref{eqn:pca} is equivalent with the method of
\term{principal component analysis} (PCA) from the statistics literature~\cite{Jol02:Principal-Component}
and the \term{total least squares} (TLS) method from the linear algebra community~\cite{VHV87:Total-Least}.

The mathematical program~\eqref{eqn:pca} is not convex because orthoprojectors do not form a convex set,
so we have no right to expect that the problem is tractable.
Nevertheless, we can compute an analytic solution by means of a singular value decomposition (SVD) of the data~\cite{EY39:Principal-Axis,VHV87:Total-Least}.
Suppose that $\mtx{X}$ is a $D \times N$ matrix whose columns are the data points,
arranged in fixed order, and let $\mtx{X} = \mtx{U\Sigma V}^\transp$ be an SVD of this matrix.  Form the $D \times d$ matrix $\mtx{U}_d$ by extracting the first $d$ columns of $\mtx{U}$; the columns of $\mtx{U}_d$ are often called the \term{principal components} of the data.  Then we can construct an optimal point $\Proj_{\star}$ for~\eqref{eqn:pca} using the formula $\Proj_{\star} = \mtx{U}_d \, {\mtx{U}_d}^\transp$.

\subsection{Classical Methods for Achieving Robustness}

Imagine now that the dataset $\Xall$ contains \emph{inliers}, points we hope to explain with a linear model, as well as \emph{outliers}, points that come from another process, such as a different population or noise.  The data are not labeled, so it may be challenging to distinguish inliers from outliers.  If we apply the PCA formulation~\eqref{eqn:pca} to fit a subspace to $\Xall$, the rogue points can interfere with the linear model for the inliers.

To guard the subspace estimation procedure against outliers,
statisticians have proposed to replace the sum of squares in~\eqref{eqn:pca}
with a figure of merit that is less sensitive to outliers.
One possibility is to sum the \emph{unsquared} residuals, which reduces the contribution from
large residuals that may result from aberrant data points.  This idea leads to
the following optimization problem.
\begin{align} \label{eqn:orlav}
\minimize \sum_{\vct{x} \in \Xall} \enorm{ \vct{x} - \Proj \vct{x} }
\subjto &\text{$\Proj$ is an orthoprojector}
\quad\textsf{and}\quad \notag \\[-10pt]
&\trace \Proj = d.
\end{align}
In case $d = D - 1$, the problem~\eqref{eqn:orlav} is sometimes called
\term{orthogonal $\ell_1$ regression}~\cite{SW87:Orthogonal-Linear}
or \term{least orthogonal absolute deviations}~\cite{Nyq88:Least-Orthogonal}.
The extension to general $d$ is apparently more recent~\cite{Wat01:Some-Problems,DZHZ06:R1-PCA}.
See the books~\cite{HR09:Robust-Statistics,RL87:Robust-Regression,MMY06:Robust-Statistics}
for an extensive discussion of other ways to combine residuals to obtain robust estimators.

Unfortunately, the mathematical program~\eqref{eqn:orlav} is not convex,
and, in contrast to~\eqref{eqn:pca}, no \lang{deus ex machina}
emerges to make the problem tractable.
Although there are many algorithms~\cite{Nyq88:Least-Orthogonal,CM91:Iterative-Linear,BH93:Orthogonal-Linear,Wat02:Gauss-Newton-Method,DZHZ06:R1-PCA,ZSL09:Median-k-Flats} that attempt~\eqref{eqn:orlav},
none is guaranteed to return a global minimum.
In fact, most of the classical proposals for robust linear modeling involve intractable optimization
problems, which makes them poor options for computation in spite of their
theoretical properties~\cite{MMY06:Robust-Statistics}.

\subsection{A Convex Program for Robust Linear Modeling}

The goal of this paper is to
develop, analyze, and test a
rigorous method for fitting
robust linear models by means
of convex optimization.  We
propose to \emph{relax} the
hard optimization problem~\eqref{eqn:orlav}
by replacing the nonconvex constraint set
with a larger convex set.
The advantage of this approach
is that we can solve the resulting convex program
completely using a variety of efficient algorithms.

The idea behind our relaxation is straightforward.
Each eigenvalue of an orthoprojector $\Proj$ equals zero or one
because $\Proj^2 = \Proj$.
Although a 0--1 constraint on eigenvalues is hard to enforce, the
symmetric matrices whose eigenvalues lie in the interval $[0, 1]$
form a convex set.
This observation leads us to frame the following convex optimization problem.
Given a dataset $\Xall$ in $\R^D$ and a target dimension $d \in \{1, 2, \dots, D - 1\}$ for the linear
model, we solve
\begin{equation} \label{eqn:rrp}
\minimize	\sum_{\vct{x} \in \Xall} \norm{ \vct{x} - \mtx{P} \vct{x} }
\subjto		\mtx{0} \psdle \mtx{P} \psdle \Id
\quad\textsf{and}\quad
\trace  \mtx{P}  = d.
\end{equation}
We refer to~\eqref{eqn:rrp} as \rrp\ because it attempts to harvest
linear structure from data.

\subsubsection{A Tighter Relaxation?}
\label{sec:stronger-relaxation}

One may wonder whether it is possible to find a tighter relaxation
of~\eqref{eqn:orlav} than our proposed formulation~\eqref{eqn:rrp}.
If we restrict our attention to convex programs, the answer is negative.

\begin{fact} \label{fact:proj-relax}
For each integer $d \in [0, D]$, the set
$\{ \mtx{P} \in \R^{D \times D} : \mtx{0} \psdle \mtx{P} \psdle \Id \text{ and }
\trace \mtx{P} = d \}$ is the convex hull of the
$D \times D$ orthoprojectors with trace $d$.
\end{fact}

\noindent
To prove this fact, it suffices to apply a diagonalization argument and
to check that the set
$\{\vct{\lambda} \in \R^D : \sum_{i=1}^D \lambda_i =d \text{ and } 0 \leq  \lambda_i \leq 1\}$
is the convex hull of
the set of vectors that have $d$ ones and $D-d$ zeros.  See~\cite{OW92:Sum-Largest}
for a discussion of this result.

Fact~\ref{fact:proj-relax} gives a geometric indication about why \rrp\ might be
effective.  Suppose there is a rank-$d$ orthoprojector $\Proj_L$ that
provides a good linear model for the inliers.  The constraint set
in~\eqref{eqn:rrp} is the convex hull of the rank-$d$ orthoprojectors.
In high dimensions, convex hulls tend to be very small, so there are
relatively few perturbations of $\Proj_L$ that remain feasible for~\eqref{eqn:rrp}.
At the same time, the objective function in~\eqref{eqn:rrp} is a sort of $\ell_1$ norm,
so it has relatively few directions of descent at $\Proj_L$.  We have the intuition
that it is impossible to move far from $\Proj_L$ into the constraint set while
simultaneously reducing the objective function.  This insight is ultimately the
basis for our analysis.

\subsubsection{Computing an Orthoprojector from the Solution of \rrp}
\label{sec:postprocessing}

It is easy to see that a solution $\mtx{P}_{\star}$ to the \rrp\ problem has rank $d$ or greater.  On the other hand, the matrix $\mtx{P}_{\star}$ does need not to be an orthoprojector, so it is not immediately clear how to obtain a $d$-dimensional linear model from a minimizer of \rrp.  To accomplish this goal, let us consider the auxiliary problem
\begin{align}
\minimize \pnorm{S_1}{ \mtx{P}_{\star} - \Proj}
\subjto
&\text{$\Proj$ is an orthoprojector}
\quad\text{and}\quad \notag \\[-4pt]
&\trace \Proj = d.\label{eqn:closest-proj}
\end{align}
In other words, we find a rank-$d$ orthoprojector $\Proj_{\star}$ that is closest to $\mtx{P}_{\star}$ in Schatten 1-norm.  We use the range of $\Proj_{\star}$ as our linear model.

It is straightforward to compute a solution $\Proj_{\star}$ to the problem~\eqref{eqn:closest-proj}.   We just need to construct an orthogonal projector whose range is a dominant $d$-dimensional invariant subspace of $\mtx{P}_{\star}$. More precisely, we form the spectral factorization
$\mtx{P}_{\star} = \mtx{U\Lambda U}^\transp$ where the entries of the diagonal matrix $\mtx{\Lambda}$ are listed in weakly decreasing order.  Extract the $D \times d$ matrix $\mtx{U}_d$ consisting of the first $d$ columns of $\mtx{U}$.  Then an optimal point for~\eqref{eqn:closest-proj} is given by the formula $\Proj_{\star} = \mtx{U}_d {\mtx{U}_d}^\transp$. (This well-known recipe for solving~\eqref{eqn:closest-proj} can be verified using a straightforward modification of the argument leading to~\cite[Thm.~IX.7.2]{Bha97:Matrix-Analysis}.)

The range of the matrix $\Proj_{\star}$ often provides
a very good fit for the inlying data points, even when
there are many outliers.  This paper
provides theoretical and empirical support for this claim.
In Section~\ref{sec:algorithm}, we present a numerical algorithm
for solving~\eqref{eqn:rrp} efficiently.  Section~\ref{sec:practical}
outlines some practical issues that are important in applications.

\subsection{Main Contributions}

This work partakes in a larger research vision:
Given a difficult nonconvex optimization problem,
it is often more effective to solve a convex variant
than to accept a local minimizer of the original problem.

We believe that the main point of interest is our application of convex optimization to solve a problem involving subspaces.  There are two key observations here.  We parameterize subspaces by orthoprojectors, and then we replace the set of rank-$d$ orthoprojectors with its convex hull.  This relaxation has a different character from previous approaches to robust linear modeling, so we have found it necessary to develop a new type of analysis
to obtain theoretical results for \rrp.

We have also done some numerical work which indicates that \rrp\ can be more effective than its competitors for certain types of robust linear modeling.  After the original version~\cite{LMTZ12:Robust-Computation} of this manuscript appeared, our ideas have been applied to a difficult class of problems involving orthogonality constraints, and this approach sometimes outperforms its competitors~\cite{WS12:Exact-Stable}.  Together, these papers suggest that relaxations like \rrp\ can be used to address important geometric questions in data analysis.

\subsection{Roadmap}

We close this introduction with an outline of the paper.
In Section~\ref{sec:cond-observ}, we develop a deterministic
analysis of \rrp\ that describes when it can recover a
linear model from a noisy dataset that includes outliers.
Section~\ref{sec:haystack} instantiates this result for
a simple random data model.  In Section~\ref{sec:algorithm},
we develop an efficient numerical method for solving the
\rrp\ problem.  Then we describe a numerical example involving
an image database in Section~\ref{sec:experiments}.
We discuss related work in Section~\ref{sec:previous}.
The technical details that support our work appear in the appendices.

\section{Theoretical Analysis of the \rrp\ Problem}
\label{sec:cond-observ}

The goal of this section is to provide theoretical
evidence that the \rrp\ problem~\eqref{eqn:rrp} is an effective way
to find a robust linear model for a dataset.
To do so, we consider a very general deterministic
setup where the data consists of inliers that are
located near a fixed subspace and outliers that
may appear anywhere in the ambient space.  We then
introduce summary statistics for the data that
encapsulate some of its geometric properties.
Using these statistics, we state our main
result, Theorem~\ref{thm:rrp-stable}, which
gives a bound on how well \rrp\ is able to
approximate the model subspace.  This result
indicates why \rrp\
may be more effective than PCA for very noisy data.
At the end of the section, we summarize the main ideas
in the proof of Theorem~\ref{thm:rrp-stable}, leaving
the remaining details until Appendix~\ref{sec:pf-rrp-lemmas}.

\subsection{A Deterministic Data Model}

To analyze the performance of the \rrp\ method, we need
to introduce a model for the input data.  It is natural
to consider the case where the dataset contains inliers
that lie on or near a fixed low-dimensional subspace, while
the outliers can be arrayed arbitrarily in the ambient
space.  We formalize this intuition in a set of assumptions
that we refer to as the In \& Out Model, and we direct the reader to
Table~\ref{tab:in-out} for a detailed list of the parameters.

\begin{table}[h!]
\begin{center}
\begin{minipage}{0.7\textwidth}
\caption{\textsl{The In \& Out Model.}  A deterministic model for
data with linear structure that is contaminated with outliers.} \label{tab:in-out}
\renewcommand{\arraystretch}{1.25}
\rowcolors{1}{white}{liteblue}
\begin{tabular}{|l|p{.9\columnwidth}|}
\hline
$D$ & Dimension of the ambient space \\
$L$ & A proper $d$-dimensional subspace of $\R^D$ \\
$N_{\In}$ & Number of inliers \\
$N_{\Out}$ & Number of outliers \\
\hline\hline
$\Xin$ & Dataset of $N_{\In}$ inliers, located ``near'' the subspace $L$ \\
$\Xout$ & Dataset of $N_{\Out}$ outliers, at arbitrary locations in $\R^D \setminus L$ \\
$\Xall$ & Dataset $\Xin \cup \Xout$ containing all the observations \\
\hline \hline
$\mtx{X}_{\Out}$ & $D \times N_{\Out}$ matrix whose columns are the outliers \\
\hline
\end{tabular}
\end{minipage}
\end{center}
\end{table}

\noindent
The key point about the In \& Out Model is that all
the inliers are located near a subspace $L$, so it is
reasonable for us to investigate when an algorithm can
approximate this target subspace $L$.

\subsection{Summary Parameters for the In \& Out Model}

The In \& Out Model is very general, so we cannot hope to
approximate the target subspace $L$ without making further
assumptions on the data.  In this section, we develop some
geometric summary statistics that allow us to check when
\rrp\ is effective at finding the subspace $L$.  Heuristically,
we need the inliers to provide a significant amount of evidence
for the subspace, while the outliers cannot exhibit too much
linear structure.  Otherwise, an unsupervised algorithm would
be justified in finding a subspace that describes the outliers
instead of the inliers!

Let us begin with a discussion of what it means for the inliers
to provide evidence for a specific subspace $M \subset \R^D$.
Imagine that we approximate each inlier $\vct{x}$
with the point $\Proj_M \vct{x}$ in the subspace $M$.
These approximations must have two properties.
First, we want the approximations of the inliers
to corroborate all the directions in the subspace $M$.
Second, we need to be sure that the residual error
in the approximations is not too large.  Our first
two summary statistics are designed to address these
requirements.

To quantify how well the inliers fill out
a subspace $M \subset \R^D$, we introduce the
\term{permeance statistic} $\Perm(M)$.
\begin{equation} \label{eqn:permeance}
\Perm(M) := \inf_{\substack{\vct{u} \in M \\ \enorm{\vct{u}} = 1}} \
	\sum_{\vct{x} \in \Xin} \absip{ \vct{u} }{ \Proj_M \vct{x} }.
\end{equation}
If there is a direction in the subspace $M$ that is
orthogonal to each inlier, then the permeance statistic $\Perm(M)$
is zero.  On the other hand, the permeance statistic is large
when every direction $\vct{u}$ in $M$ has the property
that many inliers have a component along $\vct{u}$.

Second, we introduce the \term{total inlier residual} $\Resid(M)$
to measure the total error that we incur by approximating the data
using the subspace $M$.
\begin{equation} \label{eqn:residual}
\Resid(M) := \sum_{\vct{x} \in \Xin}
	\enorm{ \Proj_{M^\perp} \vct{x} }.
\end{equation}
Let us emphasize that the total inlier residual is less sensitive
to large errors than the sum of squared residuals that drives the
PCA method.

Next, let us turn to the condition that we require of the outliers.
A major challenge for any robust linear modeling procedure is the
possibility that both the inliers and the outliers exhibit linear
structure.  In this case, an algorithm may choose to fit a linear
model to the outliers if they have a stronger signature.

To measure the amount of linear structure in the outliers, we introduce
the \term{alignment statistic} $\Align(M)$ with respect to a
target subspace $M$.
\begin{equation} \label{eqn:structure}
\Align(M) := \norm{ \mtx{X}_{\Out} } \cdot \vsmnorm{}{ \widetilde{\Proj_{M^\perp} \mtx{X}_{\Out}} },
\end{equation}
where $\mtx{X}_{\Out}$ is the matrix whose columns are the outlying
data points and the spherization operator \ $\widetilde{\,}$ \
normalizes the columns of a matrix.
It is somewhat harder to understand what the alignment
statistic $\Align(M)$ reflects.  First, observe that the spectral norm
$\norm{\mtx{X}_{\Out}}$ tends to be large when the outliers are
collinear, and it is small when the outliers are weakly correlated.
The other term in the alignment statistic asks about the collinearity
of the outliers after we have removed their components in the subspace
$M$.

Finally, we present one more statistic that weighs the influence
of the inliers against the influence of the outliers.  The
\term{stability statistic} $\Stab(M)$ of the data with respect
to a subspace $M \subset \R^D$ is the quantity
\begin{equation} \label{eqn:stability}
\Stab(M) := \frac{\Perm(M)}{4\sqrt{\dim(M)}} - \Align(M).
\end{equation}
The stability statistic tends to be large when the inliers
provide a lot of evidence for the subspace $M$ and the outliers
contain relatively little distracting linear structure.
As we will see, when $\Stab(M)$ is large, the \rrp\ method
can be very effective at approximating the subspace $M$,
even when the inliers are noisy.

\subsection{Performance of \rrp\ with Deterministic Data}

The main theoretical result in this paper describes the
behavior of the \rrp\ method when it is applied to data
that meet the assumptions of the In \& Out Model from
Table~\ref{tab:in-out}.

\begin{thm}[Performance Analysis for \rrp] \label{thm:rrp-stable}
Fix any $d$-dimensional subspace $L$ of $\R^D$, and assume that
$\Xall$ is a dataset that conforms to the In \& Out Model on
page~\pageref{tab:in-out}.  Let $\mtx{P}_{\star}$
be a solution to the \rrp\ problem~\eqref{eqn:rrp},
and find the nearest $d$-dimensional orthoprojector
$\Proj_{\star}$ by solving~\eqref{eqn:closest-proj}.
Then we have the error bound
$$
\pnorm{S_1}{ \Proj_{\star} - \Proj_{L} }
	\leq \frac{4 \, \Resid(L)}{[ \Stab(L) - \Resid(L) ]_+}.
$$
The stability statistic $\Stab(L)$ is defined in~\eqref{eqn:stability},
and the total inlier residual $\Resid(L)$ is defined in~\eqref{eqn:residual}.
\end{thm}

\noindent
An overview of the proof of Theorem~\ref{thm:rrp-stable} appears
below in Section~\ref{sec:rrp-stable-pf}.  Before we present the
argument, let us explain the content of this result.

\begin{enumerate} \setlength{\itemsep}{4pt}
\item	Assume that all the inliers are contained within the target
subspace $L$.  Then the total residual $\Resid(L) = 0$.  If the stability
statistic $\Stab(L) > 0$, then Theorem~\ref{thm:rrp-stable} ensures that
$\Proj_{\star} = \Proj_L$.  In other words, we recover the subspace $L$
without error.

\item	Again, suppose that the inliers are located within the target subspace $L$.
As we begin to move the inliers away from $L$, the error in approximating the
subspace increases at a linear rate proportional to $\Stab(L)^{-1}$.  Therefore,
when the stability statistic is large, the noise in the inliers has a very small
impact on the approximation error.

\item The effect of outliers appears only through the alignment statistic~\eqref{eqn:structure}. When the inliers lie in the subspace \(L\), the alignment statistic is the largest when the outliers cluster along a one-dimensional subspace in \(L^\perp\).  With adversarial outliers, our theory indicates that a very large permeance~\eqref{eqn:permeance}  is required to counteract linear structure in the outliers.

\item We have measured the distance between the projectors using the Schatten 1-norm, which provides a very strong bound indeed.  To appreciate the value of this type of estimate, note that it follows from~\cite[p. 202]{Bha97:Matrix-Analysis} that for any two \(d\)-dimensional subspaces \(M\), \(M'\) of \(\R^D\), \begin{equation*}
  \pnorm{S_1}{\Proj_M-\Proj_{M'}} = 2\sum_{i=1}^d \sin \theta_i(M,M') \ge  \frac{4}{\pi} \sum_{i=1}^d \theta_i(M,M'),
\end{equation*}
where $\theta_i(M, M')$ is the $i$th principal angle between the subspaces, and we use the fact that \(\sin(\theta)\ge 2\theta/\pi \) for \(0\le \theta\le \pi/2\).
Therefore, our error bound allows us to control all the principal angles
between the computed subspace $\range(\Proj_{\star})$ and the target subspace $L$.

\item	Imagine that we knew in advance which points were inliers.  Then we could pose
the oracle $\ell_1$ orthogonal regression problem:
\begin{align*}
\minimize	\sum_{\vct{x} \in \Xin} \enorm{ (\Id - \Proj) \vct{x} }
\subjto &\text{$\Proj$ is an orthoprojector}
\quad\text{and}\quad \notag \\[-10pt]
&\trace \Proj = d.
\end{align*}
Let $\Proj_{\rm oracle}$ be a solution to this (apparently intractable) problem.  Then the
subspace $L_{\rm oracle} := \range( \Proj_{\rm oracle} )$ minimizes the total inlier residual $\Resid(M)$ over $d$-dimensional subspaces $M \subset \R^D$.  When we apply
Theorem~\ref{thm:rrp-stable} with $L = L_{\rm oracle}$, we discover that \rrp\
identifies a linear model that is close to the oracle $\ell_1$ model---provided
that the oracle model is sufficiently stable.  This observation is interesting
even when there are no outliers.

\item	How does \rrp\ compare with standard PCA?  The formulation~\eqref{eqn:pca}
shows that PCA searches for a subspace by minimizing the sum of squared residuals.
On the other hand, we have just seen that \rrp\ is (almost) capable of finding a
subspace that minimizes the sum of unsquared residuals.  It is well known that
the sum of unsquared residuals tends to be much less sensitive to large errors
than the sum of squared residuals.  As a consequence, we expect that \rrp\ will be
more effective at ignoring data points that contribute large errors.  See Figure~\ref{fig:reaper-haystack-noise} below for numerical evidence of this phenomenon. 
\end{enumerate}

\noindent
In short, Theorem~\ref{thm:rrp-stable} indicates that \rrp\ has the qualitative
features that one desires in a method for robust linear modeling.
In Section~\ref{sec:haystack}, we instantiate the result for a simple
random model to offer some insight about how the summary statistics
scale.

\subsection{Proof of Theorem~\ref{thm:rrp-stable}}
\label{sec:rrp-stable-pf}

This section contains the main steps in the proof of Theorem~\ref{thm:rrp-stable}.
Most of the technical details are encapsulated in two lemmata, which we
establish in Appendix~\ref{sec:pf-rrp-lemmas}.  Throughout this section
and the appendix, we retain the notation and assumptions of the In \& Out
Model from page~\pageref{tab:in-out}.

The argument is based on several ideas.  First, if the inliers are
contained within a low-dimensional subspace $L$, then \rrp\ can
identify this subspace whenever the stability statistic $\Stab(L) > 0$.
To show that~\eqref{eqn:rrp} recovers $L$ exactly in this case, we prove
that every feasible perturbation of the projector $\Proj_L$ increases
the objective.  This type of primal analysis is similar in spirit to the
argument in~\cite{Tro06:Just-Relax,Tro09:Corrigendum-Just}, but the
technical details are harder because we are working with matrices.
It contrasts with the style of analysis that dominates recent papers
on convex methods for robust linear modeling, which are usually
based on elaborate constructions of dual certificates.

Second, when the inliers are not contained in the subspace $L$,
we can use a perturbation analysis to assess how much the noise impacts
the solution to the
optimization problem.  The key idea here is to replace the objective function
in~\eqref{eqn:rrp} with a nearby objective function.  This alteration
allows us to take advantage of the exact recovery results that we mentioned
in the last paragraph.  The approach is based on some classic arguments
in optimization; see~\cite[Sec.~4.4.1]{BS02:Perturbation-Analysis}.  We
do not believe these ideas have been applied in the literature on convex
relaxations of data analysis problems.

To begin, we introduce some notation.  The \rrp\ problem~\eqref{eqn:rrp}
can be framed as
\begin{equation} \label{eqn:rrp-proof}
\minimize f(\mtx{P})
\subjto \mtx{P} \in \Phi.
\end{equation}
with objective function
\begin{equation} \label{eqn:rrp-objective}
f( \mtx{P} ) := \sum_{\vct{x} \in \Xall} \enorm{ (\Id - \mtx{P}) \vct{x} }
\end{equation}
and  feasible set
\begin{equation} \label{eqn:rrp-feasible}
\Phi := \{ \mtx{P} : \mtx{0} \psdle \mtx{P} \psdle \Id
\quad\text{and}\quad \trace(\mtx{P}) = d \}.
\end{equation}
Let $\mtx{P}_{\star}$ be any solution to~\eqref{eqn:rrp}.
Next, we find a solution $\Proj_{\star}$ to the problem
\begin{equation} \label{eqn:nearest-proj-proof}
\minimize \pnorm{S_1}{ \mtx{P}_{\star} - \Proj }
\subjto
\text{$\Proj$ is a rank-$d$ orthoprojector}.
\end{equation}
Our aim is to compare the computed projector
$\Proj_{\star}$ with the target projector $\Proj_{L}$.

The main technical insight is to use the target projector
$\Proj_L$ to construct a perturbation $g$ of the objective function $f$
of the \rrp\ problem:
\begin{equation} \label{eqn:perturbed-objective}
g(\mtx{P} ) : = \sum_{\vct{x} \in \Xin} \enorm{ (\Id - \mtx{P}) \Proj_{L} \vct{x} } + \sum_{\vct{x} \in \Xout} \enorm{ (\Id - \mtx{P}) \vct{x} }.
\end{equation}
To perform the analysis, we pass from the original optimization
problem~\eqref{eqn:rrp-proof} to the perturbed problem
\begin{equation} \label{eqn:rrp-perturbed}
\minimize g(\mtx{P})
\subjto \mtx{P} \in \Phi.
\end{equation}
Observe that, if the inliers are contained in the target subspace $L$,
then the perturbed problem~\eqref{eqn:rrp-perturbed}
coincides with the original problem~\eqref{eqn:rrp-proof}.

The argument requires two technical results.  The first lemma shows that the total inlier residual $\Resid(L)$ controls the difference between the perturbed objective $g$ and the original objective $f$. The second lemma shows, in particular, that $\Proj_L$ is the unique minimizer of~\eqref{eqn:rrp-perturbed} when the stability statistic $\Stab(L) > 0$.  Together, these estimates allow us to conclude that the solution to the original problem~\eqref{eqn:rrp-proof} is not far from $\Proj_L$.

More precisely, we demonstrate that the perturbed objective $g$ is close to the original
objective $f$ for matrices close to $\Proj_L$.

\begin{lemma}[Controlling the Size of the Perturbation] \label{lem:control-perturb}
Introduce the difference $h := f - g$ between the two objectives.  Then
$$
\abs{ h(\Proj_{L}) - h(\Proj_{L} + \mtx{\Delta}) }
	\leq \Resid(L) \cdot \big[ 2 +  \pnorm{S_1}{\mtx{\Delta}} \big].
$$
for any symmetric matrix $\mtx{\Delta}$.  The total inlier residual $\Resid(L)$ is defined in~\eqref{eqn:residual}.
\end{lemma}

\noindent
The proof of Lemma~\ref{lem:control-perturb} appears in Appendix~\ref{sec:control-perturb}.

We also argue that that the perturbed objective function
$g$ increases quickly when we move away from the point $\Proj_{L}$ into
the feasible set.

\begin{lemma}[Rate of Ascent of the Perturbed Objective] \label{lem:rate-ascent}
Assume that $\Proj_L + \mtx{\Delta} \in \Phi$.  Then
$$
g(\Proj_{L} + \mtx{\Delta}) - g(\Proj_{L}) \geq
	\Stab(L) \cdot \pnorm{S_1}{ \mtx{\Delta} }.
$$
The stability statistic $\Stab(L)$ is defined in~\eqref{eqn:stability}.
\end{lemma}

\noindent
The proof of Lemma~\ref{lem:rate-ascent} appears in Appendix~\ref{sec:rate-ascent}.

Granted these two results, we quickly complete the proof of Theorem~\ref{thm:rrp-stable}.  Define the function $h := f - g$.  Adding and subtracting terms, we find that
\begin{equation} \label{eqn:rrp-pf-1}
g(\mtx{P}_{\star}) - g(\Proj_{L})
	= \big[ h(\Proj_{L}) - h(\mtx{P}_{\star}) \big] + \big[ f( \mtx{P}_{\star} ) - f(\Proj_{L}) \big]
	\leq h(\Proj_{L}) - h(\mtx{P}_{\star}).
\end{equation}
The inequality in~\eqref{eqn:rrp-pf-1} holds because the second bracket is nonpositive.  Indeed, $\mtx{P}_{\star}$ minimizes $f$ over the feasible set $\Phi$, and $\Proj_L$ is
also a member of the feasible set.  Set $\mtx{\Delta} = \mtx{P}_{\star} - \Proj_L$, and apply Lemmas~\ref{lem:control-perturb} and~\ref{lem:rate-ascent} to bound the right- and  left-hand sides of~\eqref{eqn:rrp-pf-1}.  We reach
$$
\Stab(L) \cdot \pnorm{S_1}{ \mtx{P}_{\star} - \Proj_{L} }
	\leq 2 \, \Resid(L) + \Resid(L) \cdot \pnorm{S_1}{ \mtx{P}_{\star} - \Proj_{L} }.
$$
Solve this inequality to reach the bound
\begin{equation} \label{eqn:rrp-almost-there}
\pnorm{S_1}{ \mtx{P}_{\star} - \Proj_{L} }
	\leq \frac{2\, \Resid(L)}{[\Stab(L) - \Resid(L)]_+}.
\end{equation}
To finish the argument, note that
$$
\pnorm{S_1}{ \Proj_{\star} - \Proj_{L} }
	\leq \pnorm{S_1}{ \Proj_{\star} - \mtx{P}_{\star} }
	+ \pnorm{S_1}{ \mtx{P}_{\star} - \Proj_{L} }
	\leq 2 \pnorm{S_1}{ \mtx{P}_{\star} - \Proj_{L} }
	\leq \frac{4 \, \Resid(L)}{[\Stab(L)-\Resid(L)]_+}.
$$
The first bound follows from the triangle inequality.
The second estimate holds because the distance from
$\Proj_{\star}$ to $\mtx{P}_{\star}$ is no greater than
the distance from $\Proj_{L}$ to $\mtx{P}_{\star}$
because $\Proj_{\star}$ is a minimizer of~\eqref{eqn:nearest-proj-proof}.
The last inequality follows from~\eqref{eqn:rrp-almost-there}.

\section{Theoretical Example: The Haystack Model}
\label{sec:haystack}

The In \& Out Model is very general, so Theorem~\ref{thm:rrp-stable} applies to a wide variety of specific examples.  To see the kind of results that are possible, let us apply Theorem~\ref{thm:rrp-stable} to study the behavior of \rrp\ for data drawn from a simple random model.  We use standard tools from high-dimensional probability to compute the values of the summary statistics.

\subsection{The Haystack Model}

Let us consider a simple generative random model for a dataset.
We call this the Haystack Model, and we refer the reader to
Table~\ref{tab:haystack} for a list of the assumptions and the parameters.
The Haystack Model is not intended as a realistic description
of data.  Instead, the goal is to capture the idea that
inliers admit a low-dimensional linear model, while the
outliers are totally unstructured.

\begin{table}[h!]
\begin{center}
\begin{minipage}{0.7\textwidth}
\caption{\textsl{The Haystack Model.}  A generative random model for
data with linear structure that is contaminated with outliers. The abbreviation \term{i.i.d.}  stands for \emph{independent and identically distributed}. } \label{tab:haystack}
\renewcommand{\arraystretch}{1.25}
\rowcolors{1}{white}{liteblue}
\textbf{Parameters}\\[2pt]
\begin{tabular}{|p{0.068\columnwidth}|p{.9\columnwidth}|}
\hline
$D$ & Dimension of the ambient space \\
$L$ & A proper $d$-dimensional subspace of $\R^D$ containing the inliers \\
$N_{\In}$ & Number of inliers \\
$N_{\Out}$ & Number of outliers \\
$\rho_{\In}$ & Inlier sampling ratio $\rho_{\In} := N_{\In} / d$ \\
$\rho_{\Out}$ & Outlier sampling ratio $\rho_{\Out} := N_{\Out} / D$ \\
$\sigma_{\In}^2$ & Variance of the inliers per subspace dimension \\
$\sigma_{\Out}^2$ & Variance of the outliers per ambient dimension \\
\hline
\end{tabular}
~\\
\textbf{Data}\\[2pt]
\begin{tabular}{|p{0.068\columnwidth}|p{.9\columnwidth}|}
\hline
$\Xin$ & Set of $N_{\In}$ inliers, drawn i.i.d.~$\normal(\vct{0},\ (\sigma^2_{\In} / d) \, \Proj_L)$ \\
$\Xout$ & Set of $N_{\Out}$ outliers, drawn i.i.d.~$\normal(\vct{0},\ (\sigma^2_{\Out} / D) \, \Id_D)$ \\
$\Xall$ & The set $\Xin \cup \Xout$ containing all the data points \\
\hline
\end{tabular}
\end{minipage}
\end{center}\noindent
\end{table}

There are a few useful intuitions associated with this model.
As the inlier sampling ratio $\rho_{\In}$ increases, the inliers fill out the subspace
$L$ more completely so the linear structure becomes more evident.  As the outlier
sampling ratio $\rho_{\Out}$ increases, the outliers become more distracting and
they may even start to exhibit some linear structure due to chance.
Next, observe that we have scaled the points so that their energy does not depend on
the dimensional parameters:
$$
\Expect \enormsq{ \vct{x} } = \sigma_{\In}^2
\quad\text{for $\vct{x} \in \Xin$}
\quad\text{and}\quad
\Expect \enormsq{ \vct{x} } = \sigma_{\Out}^2
\quad\text{for $\vct{x} \in \Xout$}.
$$
As a result, when $\sigma_{\In}^2 = \sigma_{\Out}^2$, we cannot screen outliers
just by looking at their energy.
The sampling ratios and the variances contain most of the
information about the behavior of this model.

\subsection{Analysis of the Haystack Model}

Using methods from high-dimensional probability, we can
analyze the stability statistic $\Stab(L)$ for a dataset
drawn at random from the Haystack Model.

\begin{thm}[Analysis of the Haystack Model] \label{thm:random-rrp}
Fix a number $\beta > 0$, and assume that $1 \leq d \leq (D - 1)/2$.
Let $L$ be an arbitrary $d$-dimensional subspace of $\R^D$, and draw the dataset
$\Xall$ at random according to the Haystack Model on page~\pageref{tab:haystack}.
The stability statistic satisfies the bound
\begin{equation*}
  \Stab(L) \geq \frac{\sigma_\In}{\sqrt{32\pi}} \left[ \rho_\In - \pi(4 + 2\beta) \right]
  - 6 \sigma_{\Out} \left[ \rho_{\Out} + 1 + \beta \right],\label{eq:1}
\end{equation*}
except with probability $3.5 \econst^{-\beta d}$.
\end{thm}

\noindent
The proof of Theorem~\ref{thm:random-rrp} appears in
Appendix~\ref{sec:proof-main-bound}.  The restriction \(d\le (D-1)/2\) above simplifies the  result; see Theorem~\ref{thm:gaussian-outliers}  for a comprehensive statement valid for \(1\le d\le D-1\).

To appreciate what this result means, it is helpful to set
$\sigma_{\In} = \sigma_{\Out} = 1$ and to suppress the values
of the constants:
\begin{equation}\label{eq:stab-bd-simple}
\Stab(L) \geq C_{\In}  \rho_\In
	- C_{\Out}\rho_{\Out}- C_{\beta} .
  \end{equation}
We see that the stability statistic grows linearly with the inlier sampling ratio,
and it decreases linearly with the outlier sampling ratio.

Since the inliers in the Haystack Model are contained in the subspace $L$,
Theorem~\ref{thm:rrp-stable} shows that \rrp\ recovers $L$ perfectly when the
stability statistic is positive.  Therefore, a sufficient condition
for exact recovery is that $\rho_{\In}$,
the number of inliers \emph{per subspace dimension},
should be at least a constant multiple of $\rho_{\Out}$,
the number of outliers \emph{per ambient dimension}.
As a consequence, we can find low-dimensional
linear structure in a high-dimensional space given
a small number of examples, even when the number of
outliers seems exorbitant.

\paragraph{Numerical experiment} Figure~\ref{fig:reaper-haystack} displays the results of a numerical experiment for \rrp\ under the Haystack Model. We fix the ambient dimension $D=100$ and take \(L\) a subspace of dimension $d = 1$  or $d=10$. The number $N_{\In}$ of inliers and the number $N_{\Out}$ of outliers vary over an equally-spaced\footnote{In both figures, \(N_{\Out}\) increases in increments of twenty, while \(N_{\In}\) increases in increments of two.} grid.   Note that the specific choice of the subspace $L$ is immaterial because the model is rotationally invariant.  The variance parameters are fixed $(\sigma_{\In}^2 = \sigma_{\Out}^2 = 1)$.

\begin{figure}[t]
  \centering
  \includegraphics[width=0.65\columnwidth]{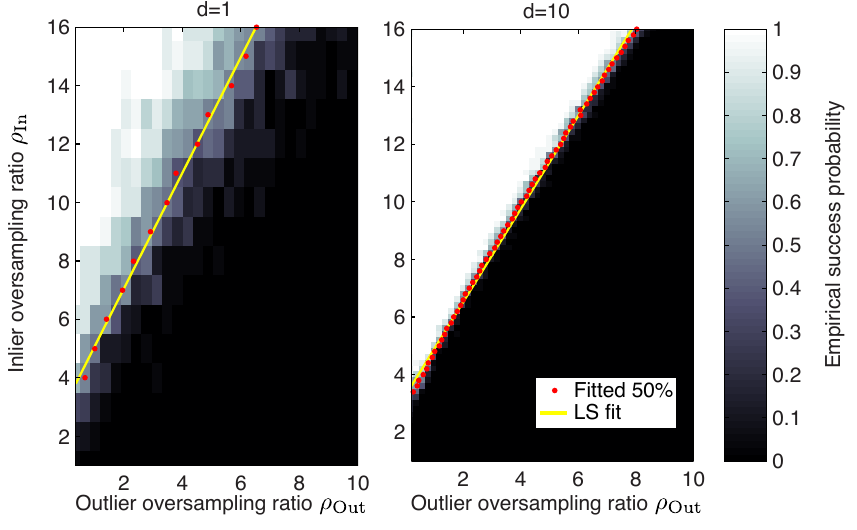}
  \caption{\textsl{Exact subspace recovery with \rrp.} The heat maps show the empirical probability that \rrp\ identifies a target subspace under the Haystack Model with varying inlier \(\rho_\In\) and outlier \(\rho_\Out\) oversampling ratios.   We perform the experiments in ambient dimension \(D=100\) with inlier dimension \(d=1\) {\sl (left)} and \(d=10\) {\sl (right)}.  For each value of \(\rho_\In\), we find the \(50\%\) empirical success \(\rho_\Out\) {\sl (red dots)}.  The yellow line indicates the least-squares fit to these points. In this parameter regime, a linear trend is clearly visible, which suggests that~\eqref{eq:stab-bd-simple} captures the qualitative behavior of \rrp\ under the Haystack Model.  }
  \label{fig:reaper-haystack}
\vspace{-5mm}
\end{figure}

We find $\mtx{P}_{\star}$  by solving
\rrp\ \eqref{eqn:rrp} with the algorithm described in Section~\ref{sec:algorithm} below, and then determine the orthoprojector \(\Proj_\star\) using~\eqref{eqn:closest-proj}.
We assess whether this procedure identifies the true subspace $\Proj_L$ subspace by  declaring the experiment a success when the error
$\pnorm{S_1}{ \mtx{P}_{\star} - \Proj_L } < 10^{-5}$.  For each pair
$(\rho_{\In}, \rho_{\Out})$, we repeat the experiment 25 times
and calculate an empirical success probability.   For each value of \(\rho_\In\), we find the \(50\%\) empirical success \(\rho_\Out\)  using a logistic fit.  We fit a line  to these points using standard least-squares.  These results indicate that the linear trend suggested by the theoretical bound~\eqref{eq:stab-bd-simple}  reflects the empirical behavior of \rrp.

\subsubsection{Noisy inliers}
To understand how \rrp\ behaves when the inlying set \(\Xin\) does not lie precisely within the target subspace, we introduce the \emph{Noisy} Haystack Model.  This model expands the standard Haystack Model from Table~\ref{tab:haystack} with the additional parameter \(\sigma_\Noise^2\) that controls the amount of noise present in the inliers.  In this extended model, the inlying data \(\Xin\) is  given by
\begin{equation}\label{eq:noise-inliers}
  \Xin\colon\text{ Set of \(N_\In\) inliers drawn i.i.d. } \normal\left(\zerovct , \frac{\sigma_{\In}^2}{d} \Proj_{L} + \frac{\sigma_{\Noise}^2 }{D-d} \Proj_{L^\perp}\right).
\end{equation}
All other parameters and data agree with the Haystack Model of Table~\ref{tab:haystack}.

The definition~\eqref{eq:noise-inliers} of the inlying data ensures that the stability statistic~\(\Stab(L)\) has the same distribution under the Noisy Haystack Model as under the plain Haystack Model. In particular, the relationship~\eqref{eq:stab-bd-simple} holds under the Noisy Haystack Model.  On the other hand,  the inlier residual statistic \(\Resid(L)\)~\eqref{eqn:residual} is not equal to zero under the noisy model, but rather satisfies
\begin{equation*}
  \Expect [\Resid(L)]  = \sum_{i=1}^{N_\In} \Expect[\enorm{\vct g_i}] \le N_\In \Expect[\enorm{\vct g_1}^2] ^{1/2}
= \sigma_\Noise N_\In,
\end{equation*}
where \(\vct g_i\sim \normal\bigl(\zerovct, (\sigma^2_\Noise/(D-d)) \Proj_{L^\perp}\bigr)\).  The inequality is Jensen's, and the last expression uses the fact that the squared norm of a Gaussian random variable on the \((D-d)\)-dimensional subspace  is \(D-d\).

A basic  concentration result indicates that the residual statistic will not exceed its mean by more than a factor of, say, two with overwhelming probability. (This claim is easily made precise using the result~\cite[Thm.~1.7.6]{Bogachev1998}.)  Combining this observation with~\eqref{eq:stab-bd-simple} and Theorem~\ref{thm:rrp-stable}, we see that  with high probability
\begin{equation*}
  \pnorm{S_1}{\Proj_\star - \Proj_L} \le \frac{2 N_\In \times\SNR}{\bigl[C_\In \rho_\In  - C_\Out \rho_\Out  - C_\beta- 2 N_\In\times  \SNR \bigr]_+}
\end{equation*}
where we define the signal-to-noise ratio \(\SNR:=\sigma_\In/\sigma_\Noise\).  This inequality suggests that \rrp\ is stable under the Noisy Haystack Model in the regime where the stability statistic \(\Stab(L)=O(1)\) and the signal-to-noise ratio \(\SNR= O( N_\In^{-1})\).  Our numerical experience suggests that this SNR restriction is conservative.

\paragraph{Numerical experiment}
\label{sec:numerical-experiment}

 Figure~\ref{fig:reaper-haystack-noise} compares the results of a numerical experiment under the Noisy Haystack Model using both \rrp\ and PCA.  As in the experiment for the basic Haystack Model, we set \(D=100\) and perform the experiment for a linear subspace \(L\) of dimension \(d=10\)  and \(d = 1\).  The variance parameters are \(\sigma_\In = \sigma_\Out = 1\), and we fix  \(\mathrm{SNR} = \sigma_\In/\sigma_\Noise=10\).  For each equally-spaced value\footnote{In this experiment, \(N_{\In}\) increases in increments of two while \(N_{\Out}\) increases in increments of \(20\).} of \(N_\In\) and \(N_\Out\), we draw the data \(\Xall\) from the Noisy Haystack Model.   We determine a projector \(\Proj_{\star}\) by solving \rrp~\eqref{eqn:rrp} and finding the closest subspace~\eqref{eqn:closest-proj}, and then we compute the error \(\pnorm{S_1}{\Proj_\star - \Proj_L}\).  We determine the same statistic for the projection given by PCA~\eqref{eqn:closest-proj}. We repeat this experiment \(25\) times for each value of \((\rho_\In,\rho_\Out)\).

The heat map in Figure~\ref{fig:reaper-haystack-noise} shows the mean error \(\pnorm{S_1}{\Proj_\star - \Proj_L}\)  over these trials for both \rrp\ and PCA.  The blue region of the heat map begins where the error is less than \(10\%\) of the maximum possible error
\begin{equation*}
\mathrm{MaxError} := \max_{\dim(M)=d} \pnorm{S_1}{\Proj_M-\Proj_{L}} = 2d.
\end{equation*}
We see that \rrp\ is in the blue region over more of the parameter regime than PCA, which indicates that \rrp\ is more stable than PCA under the Noisy Haystack Model.

\begin{figure}[t]
  \centering

\includegraphics[width=0.65\columnwidth]{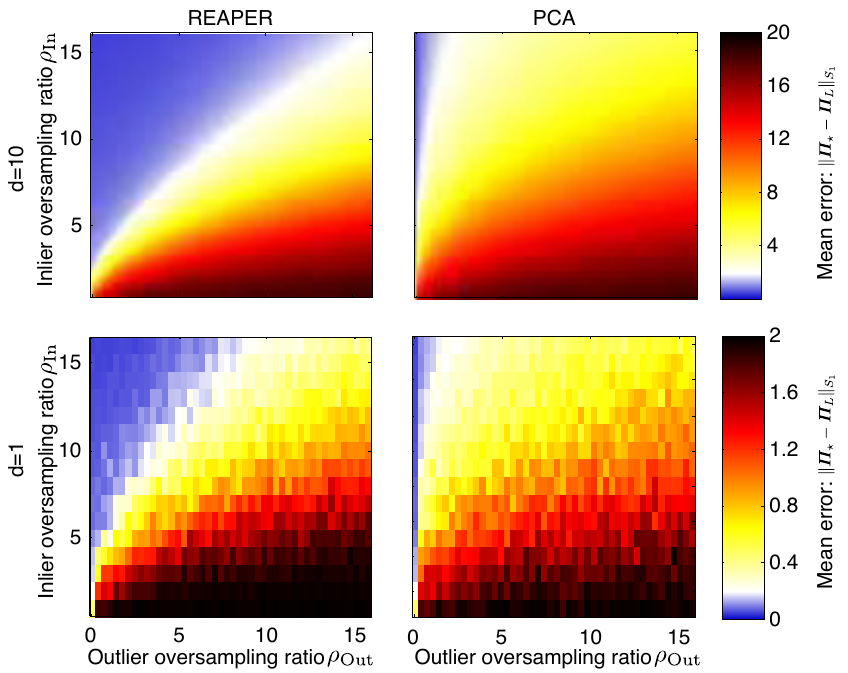}
  \caption{\textsl{Approximate subspace recovery with \rrp\ and \textnormal{PCA}.}  The heat maps show the mean error \(\pnorm{S_1}{\Proj_\star -\Proj_L}\) for the projection computed by \rrp\ and PCA.  The ambient dimension is  \(D=100\), and we perform the experiment for both \(d=10\) {\sl(top)} and \(d=1\) {\sl(bottom)}. The blue region indicates where the mean error  is less than \(10\%\) of the maximum possible error.}
  \label{fig:reaper-haystack-noise}
\vspace{-4mm}
\end{figure}
\vspace{-5mm}
\section{An Iterative Reweighted Least-Squares Algorithm for \rrp}
\label{sec:algorithm}
\vspace{-2mm}
Theorem~\ref{thm:rrp-stable} suggests that the \rrp\
problem~\eqref{eqn:rrp} can be a valuable tool for
robust linear modeling.  On the other hand,
\rrp\ is a semidefinite program, so it may
be prohibitively expensive to solve using the
standard interior-point methods.  If we intend \rrp\
to be a viable approach for data analysis problems,
it is incumbent that we produce a numerical method
with more favorable scaling properties.

In this section, we describe a numerical algorithm for
solving the \rrp\ problem~\eqref{eqn:rrp}.  Our approach
is based on the iterative reweighted least squares (IRLS)
framework~\cite[Sec.~4.5.2]{Bjo96:Numerical-Methods}.
At each step of the algorithm, we solve a weighted
least-squares problem whose weights evolve as the
algorithm proceeds.  This subproblem admits a closed-form
solution that we can obtain from a single SVD of the (weighted) data.
The IRLS method exhibits linear convergence in practice, so it can
achieve high accuracy without a substantial number of
iterations.

\vspace{-7mm}
\subsection{Solving \rrp\ via IRLS}
\vspace{-2mm}
IRLS is based on the idea that we can solve many types of weighted
least-squares problems efficiently.  Therefore, instead of
solving the \rrp\ problem~\eqref{eqn:rrp} directly, we replace
it with a sequence of weighted least-squares problems.

To motivate the approach, suppose we have an estimate $\beta_{\vct{x}} \approx \enorm{ \vct{x} - \mtx{P}_{\star} \vct{x} }^{-1}$
for each $\vct{x} \in \Xall$.  Then the \rrp\ objective at \(\mtx P_\star\) satisfies
\begin{equation*}
\smash{\sum_{\vct{x} \in \Xall}} \enorm{ \vct{x} - \mtx{P}_{\star} \vct{x} }	\approx \smash{\sum_{\vct{x} \in \Xall}} \beta_{\vct{x}} \enormsq{ \vct{x} - \mtx{P}_{\star} \vct{x} },%
\end{equation*}
and so it seems plausible that the minimizer of the following quadratic program is close to \(\mtx P_\star\).
\begin{equation} \label{eqn:irls-obj}
\minimize \smash{\sum_{\vct{x} \in \Xall}} \beta_{\vct{x}} \enormsq{ \vct{x} - \mtx{P}\vct{x} }
\subjto \mtx{0} \psdle \mtx{P} \psdle \Id
\quad\textsf{and}\quad
\trace \mtx{P} = d.
\end{equation}
We can efficiently solve problem~\eqref{eqn:irls-obj} by performing
a spectral computation and a  water-filling step that ensures  \(\zerovct \psdle \mtx P\psdle \Id\).  
The water-filling step differentiates the new algorithm from the earlier work~\cite{ZL11:Novel-M-Estimator}.
The details appear in a box labeled Algorithm~\ref{alg:weighted-ls}, and a proof of correctness appears in Appendix~\ref{sec:proof-of-subproblem-claim}.

\begin{algorithm}[b!]
\caption{{\sl Solving the weighted least-squares problem~\eqref{eqn:irls-obj}}} \label{alg:weighted-ls}
\begin{flushleft}
\footnotesize
\textsc{Input:}
\vspace{-1mm}
\begin{itemize}
\setlength{\itemsep}{0pt}
\setlength{\parskip}{0pt}
\setlength{\parsep}{0pt}
	\item	A dataset $\Xall$ of observations in $\R^D$
	\item	A nonnegative weight $\beta_{\vct{x}}$ for each $\vct{x} \in \Xall$
	\item	The dimension parameter $d$ in~\eqref{eqn:irls-obj}, where $d \in \{1, 2, \dots, D-1\}$
\end{itemize}
\textsc{Output:}
\vspace{-1mm}
\begin{itemize}
\setlength{\itemsep}{0pt}
\setlength{\parskip}{0pt}
\setlength{\parsep}{0pt}
	\item  A $D \times D$ matrix $\mtx{P}_{\star}$ that solves~\eqref{eqn:irls-obj}
\end{itemize}
\textsc{Procedure:}
\vspace{-1mm}
\renewcommand{\theenumi}{\footnotesize \arabic{enumi} \ }
\renewcommand{\labelenumi}{\theenumi}
\renewcommand{\theenumii}{\footnotesize \alph{enumii} \ }
\renewcommand{\labelenumii}{\theenumii}
\renewcommand{\theenumiii}{\footnotesize \roman{enumiii} \ }
\renewcommand{\labelenumiii}{\theenumiii}
\begin{enumerate} \setlength{\itemsep}{2pt}
\setlength{\parskip}{0pt}
\setlength{\parsep}{0pt}
	\item\label{item:weight-covar}
	Form the $D \times D$ weighted covariance matrix
	\begin{equation*}
		\mtx{C} \leftarrow \sum_{\vct{x} \in \Xall} \beta_{\vct{x}} \, \vct{xx}^\transp
	\end{equation*}
	\item\label{item:eig-decomp}
	Compute an eigenvalue decomposition
		$\mtx{C} = \mtx{U} \cdot \diag(\lambda_1, \dots, \lambda_D) \cdot \mtx{U}^\transp$
		with eigenvalues in nonincreasing order: $\lambda_1 \geq \dots \geq \lambda_D \geq 0$
	\item	\textbf{if} $\lambda_{ d + 1} = 0$ \textbf{then} %
      \begin{enumerate}  
        \setlength{\itemsep}{3pt}
        \setlength{\parskip}{0pt}
        \setlength{\parsep}{0pt}
		\item 	\textbf{for} $i = 1, \dots, D$ \textbf{do}
		\begin{equation*}
		\nu_i \leftarrow \begin{cases}
			1, & i \leq d  \\
			0, & \text{otherwise}
		\end{cases}
		\end{equation*}	
	\end{enumerate}
	\item	\textbf{else} %
	\begin{enumerate} 
      \setlength{\itemsep}{2pt}        
      \setlength{\parskip}{0pt}
        \setlength{\parsep}{0pt}
		\item	\textbf{for} $i = d + 1, \dots, D$ \textbf{do}
		\begin{enumerate}  \setlength{\itemsep}{5pt}
		\item	Set
		$$
		\theta \leftarrow \frac{i - d}{\sum\nolimits_{k=1}^i \lambda_k^{-1}}
		$$
		\item	\textbf{if} $\lambda_i > \theta \geq \lambda_{i+1}$ \textbf{then break} \textit{for}
	\end{enumerate}

	\item	\textbf{for} $i = 1, \dots, D$ \textbf{do}
	$$
	\nu_i \leftarrow \begin{cases}
		1 - \frac{\theta}{ \lambda_i }, & \lambda_i > \theta \\
		0, & \lambda_i \leq \theta
	\end{cases}
	$$
	\end{enumerate}

	\item	\textbf{return} $\mtx{P}_{\star} := \mtx{U} \cdot \diag(\nu_1, \dots, \nu_D) \cdot \mtx{U}^\transp$
\end{enumerate}
\end{flushleft}
\end{algorithm}
The heuristic above motivates an iterative procedure for solving~\eqref{eqn:rrp}.
Let $\delta$ be a (small) positive regularization parameter.
Initialize the iteration counter $k \leftarrow 0$ and the weights $\beta_{\vct{x}} \leftarrow 1$
for each $\vct{x} \in \Xall$.
We solve~\eqref{eqn:irls-obj} with the weights $\beta_{\vct{x}}$ to obtain a matrix $\mtx{P}^{(k)}$,
and then we update the weights according to the formula
$$
\beta_{\vct{x}} \leftarrow
	\frac{1}{ \max\left\{ \delta, \ \vsmnorm{}{ \vct{x} - \mtx{P}^{(k)} \vct{x} } \right\} }
\quad\text{for each $\vct{x} \in \Xall$.}
$$
In other words, we emphasize the observations that are explained well by the current
model.  The presence of the regularization parameter $\delta$ ensures that no single point can
gain undue influence.  We increment $k$, and we repeat the process until it has converged.
See the box labeled Algorithm~\ref{alg:IRLS} for the details.

The following result shows that Algorithm~\ref{alg:IRLS} is guaranteed
to converge to a point whose value is close to the optimal value of the
\rrp\ problem~\eqref{eqn:rrp}.

\begin{algorithm}[b!]
\caption{{\sl IRLS algorithm for solving the \rrp\ problem~\eqref{eqn:rrp}}} \label{alg:IRLS}
\begin{flushleft}
\textsc{Input:}
\begin{itemize}
\setlength{\itemsep}{0pt}
\setlength{\parskip}{0pt}
\setlength{\parsep}{0pt}
	\item	A dataset $\Xall$ of observations in $\R^D$
	\item	The dimension parameter $d$ in~\eqref{eqn:rrp}, where $d \in \{1, 2, \dots, D - 1\}$
	\item	A regularization parameter $\delta > 0$
	\item	A stopping tolerance $\eps > 0$
\end{itemize}
\textsc{Output:}
\begin{itemize}
\setlength{\itemsep}{0pt}
\setlength{\parskip}{0pt}
\setlength{\parsep}{0pt}
	\item  A $D \times D$ matrix $\mtx{P}_{\star}$ that satisfies $\mtx{0} \psdle \mtx{P}_{\star} \psdle \Id$
	and $\trace \mtx{P}_{\star} = d$
\end{itemize}
\textsc{Procedure:}
\renewcommand{\theenumi}{\footnotesize \arabic{enumi} \ }
\renewcommand{\labelenumi}{\theenumi}
\renewcommand{\theenumii}{\footnotesize \alph{enumii} \ }
\renewcommand{\labelenumii}{\theenumii}
\begin{enumerate} \setlength{\itemsep}{5pt}
\setlength{\parskip}{0pt}
\setlength{\parsep}{0pt}
	\item Initialize the variables:
	\begin{enumerate}
		\item Set the iteration counter $k \leftarrow 0$
		\item Set the initial error $\alpha^{(0)} \leftarrow +\infty$
		\item Set the weight $\beta_{\vct{x}} \leftarrow 1$ for each $\vct{x} \in \Xall$
	\end{enumerate}
	\item \textbf{do}
	\begin{enumerate}
		\item Increment $k \leftarrow k+1$
		\item \label{step:set_optimal}
		Use Algorithm~\ref{alg:weighted-ls} to compute an optimal point $\mtx P^{(k)}$ of~\eqref{eqn:irls-obj}
		with weights $\beta_{\vct{x}}$
		\item Let $\alpha^{(k)}$ be the optimal value of \eqref{eqn:irls-obj} at $\mtx{P}^{(k)}$
		\item Update the weights:
		\begin{equation*}
			\beta_{\vct{x}} \leftarrow \frac{1}{\max\left\{\delta, \ \vsmnorm{}{\vct{x} -\mtx P^{(k)}\vct x} \right\}}
			\quad\text{for each $\vct{x} \in \Xall$}
    	\end{equation*}
    \end{enumerate}
    \textbf{until} the objective fails to decrease: $\alpha^{(k)}\geq\alpha^{(k-1)}-\eps$
	\item \textbf{Return} $\mtx{P}_{\star} = \mtx{P}^{(k)}$
\end{enumerate}
\end{flushleft}
\end{algorithm}

\begin{thm}[Convergence of IRLS] \label{thm:IRLS-conv}
Assume that the set $\Xall$ of observations does not lie in
the union of two strict subspaces of $\R^D$.  Then the iterates
of Algorithm~\ref{alg:IRLS} with $\eps = 0$ converge to a point
$\mtx{P}_{\delta}$ that satisfies the constraints of the \rrp\
problem~\eqref{eqn:rrp}.  Moreover, the objective value at $\mtx{P}_{\delta}$
satisfies the bound
$$
\sum_{\vct{x} \in \Xall} \norm{ \vct{x} - \mtx{P}_{\delta} \vct{x} }
	- \sum_{\vct{x} \in \Xall} \norm{ \vct{x} - \mtx{P}_{\star} \vct{x} }
	\leq \frac{1}{2} \delta \abs{\Xall},
$$
where $\mtx{P}_{\star}$ is an optimal point of \rrp.
\end{thm}

\noindent
The proof of Theorem~\ref{thm:IRLS-conv} is similar to established
convergence arguments~\cite[Thms.~11 and~12]{ZL11:Novel-M-Estimator},
which follow the schema
in~\cite{CM99:Convergence-Lagged,VE80:Linear-Convergence}.
See Appendix~\ref{sec:irls-analysis} for a summary of the
proof.

\subsection{Computational Costs for Algorithm~\ref{alg:IRLS}}

Let us take a moment to summarize the computational costs for
the IRLS method, Algorithm~\ref{alg:IRLS}.  When reading through
this discussion, keep in mind that linear modeling problems typically
involve datasets where the number $N$ of data points is somewhat
larger than the ambient dimension $D$.

The bulk of the computation in Algorithm~\ref{alg:IRLS}
occurs when we solve the subproblem in Step 2b using
the weighted-least squared method from Algorithm~\ref{alg:weighted-ls}.
The bulk of the computation in Algorithm~\ref{alg:weighted-ls}
takes place during the spectral calculation in Steps~1 and~2.
In general, we need $\bigOh(ND^2)$ arithmetic operations to form the weighted
covariance matrix, and the spectral calculation requires $\bigOh(D^3)$.
The remaining steps of both algorithms have lower order.

In summary, each iteration of Algorithm~\ref{alg:IRLS}
requires $\bigOh(ND^2)$ arithmetic operations.  The algorithm converges
linearly in practice, so we need $\bigOh( \log(1/\eta) )$ iterations
to achieve an error of $\eta$.

In the statement of Algorithm~\ref{alg:weighted-ls},
we have presented the weighted least-squared calculation
in the most direct way possible.
In practice, it is usually more efficient
to form a $D \times N$ matrix $\mtx{W}$ with columns $\sqrt{\beta_{\vct{x}}} \, \vct{x}$
for $\vct{x} \in \Xall$,
to compute a thin SVD $\mtx{W} = \mtx{U\Sigma V}^\transp$,
and to set $\mtx{\Lambda} = \mtx{\Sigma}^2$.
This approach is also more stable.  In some situations, such as when $\mtx{C}$ can be guaranteed to be low rank at each iteration,
it is possible to accelerate the spectral
calculations using randomized dimension reduction as in~\cite{HMT11:Finding-Structure}.

\subsection{Empirical Convergence Rate of Algorithm~\ref{alg:IRLS}}
\label{sec:an-iter-algor}
\begin{figure}[t!]
\centering
\includegraphics[width=0.8\columnwidth]{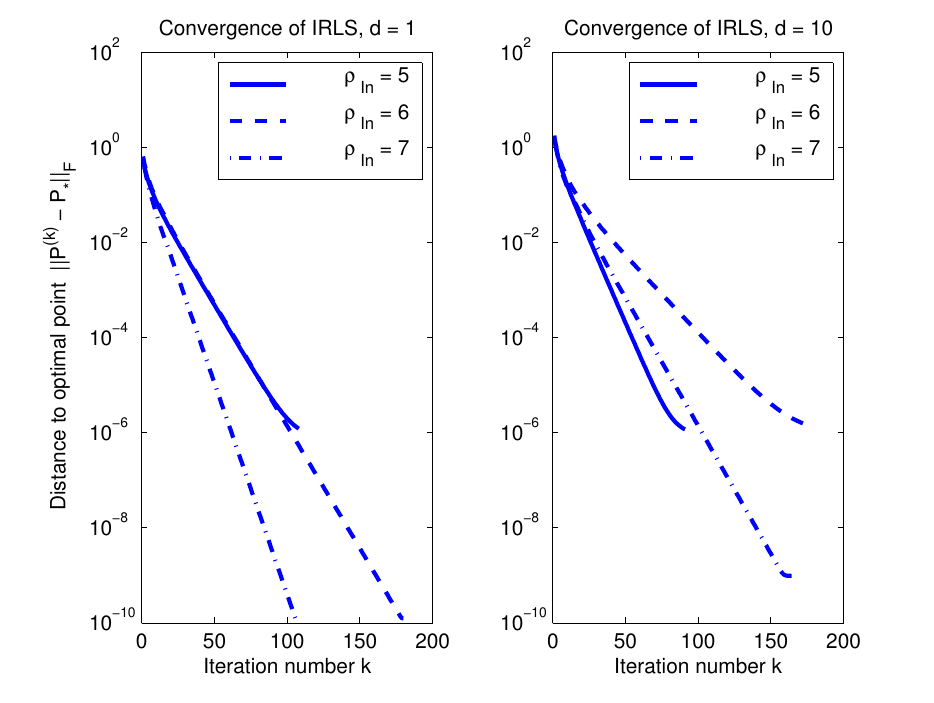}
\caption{\textsl{Convergence of IRLS to an optimal point.}
The data are drawn from the Haystack Model on page~\pageref{tab:haystack}
with ambient dimension $D = 100$ and $N_{\Out} = 200$ outliers.
Each curve is associated with a particular choice of model dimension $d$
and inlier sampling ratio $\rho_{\In} = N_{\In} / d$.  We use
Algorithm~\ref{alg:IRLS} to compute a sequence $\{\mtx{P}^{(k)}\}$
of iterates, which we compare to an optimal point $\mtx{P}_{\star}$
of the \rrp\ problem~\eqref{eqn:rrp}.  See the text in
Section~\ref{sec:an-iter-algor} for more details of the experiment.}
\label{fig:irls-lin-conv}
\end{figure}

Many algorithms based on IRLS exhibit linear convergence~\cite{CM99:Convergence-Lagged}.
Under some additional assumptions, we can prove that Algorithm~\ref{alg:IRLS}
generates a sequence $\{\mtx{P}^{(k)}\}$ of iterates that converges
linearly to an optimal point $\mtx{P}_{\star}$ of \rrp.  This argument
has a small quotient of novelty relative to the amount of technical
maneuver required, so we have chosen to omit the details.
Our analysis does not provide a realistic estimate for the rate
of convergence, so we have undertaken some numerical investigations
to obtain more insight.

Figure~\ref{fig:irls-lin-conv} indicates that, empirically,
Algorithm~\ref{alg:IRLS} does exhibit linear convergence.
In this experiment, we have drawn the data from the Haystack Model on
page~\pageref{tab:haystack} with ambient dimension $D = 100$
and $N_{\Out} = 200$ outliers.  Each curve marks the performance of a single run of Algorithm~\ref{alg:IRLS} with a unique choice of the model dimension $d$ and the inlier sampling ratio $\rho_{\In} = N_{\In} / d$. 
For this plot, we run Algorithm~\ref{alg:IRLS} with the regularization
parameter $\delta = 10^{-10}$ and the error tolerance $\eps = 10^{-15}$
to obtain a sequence $\{ \mtx{P}^{(k)} \}$ of iterates.
We compare the computed iterates with an optimal point $\mtx{P}_{\star}$
of the \rrp\ problem~\eqref{eqn:rrp} obtained by solving \rrp\
with the Matlab package \texttt{CVX}~\cite{GB08:Graph-Implementations,GB10:CVX-Matlab}
at the highest-precision setting.  The error is measured in Frobenius norm.

For synthetic data, the number of iterations required for
Algorithm~\ref{alg:IRLS} seems to depend on the difficulty of the
problem instance.  Indeed, it may take as many as 200 iterations for
the method to converge on challenging examples.
In experiments with natural data, we usually obtain good
performance after 20 iterations or so.
In practice, Algorithm~\ref{alg:IRLS} is also much faster than
algorithms~\cite{XCS11:Robust-PCA,MT11:Two-Proposals}
for solving the low-leverage decomposition
problem~\eqref{eqn:lld}.

The stopping criterion in Algorithm~\ref{alg:IRLS}
is motivated by the fact that the objective value
must decrease in each iteration.  This result is
a consequence of the proof of Theorem~\ref{thm:IRLS-conv};
see~\eqref{eqn:monotone}.
Taking $\eps$ on the order of machine precision ensures that
the algorithm terminates when the iterates are dominated by numerical
errors.  In practice, we achieve very precise results when
$\eps = 10^{-15}$ or even $\eps = 0$.  In many applications,
this degree of care is excessive, and we can obtain a
reasonable solution for much larger values of $\eps$.

\section{Numerical Example: \rrp\ Applied to an Image Database}
\label{sec:experiments}

In this section, we present a numerical experiment that
describes the performance of \rrp\ for a stylized problem
involving natural data.

\subsection{Some Practical Matters}
\label{sec:practical}

Although the \rrp\ formulation is effective on its own,
we can usually obtain better linear models if we preprocess
the data before solving~\eqref{eqn:rrp}.  Let us summarize
the recommended procedure, which appears as Algorithm~\ref{alg:proto}.

\begin{algorithm}[b!]
  \caption{{\sl Prototype algorithm for robust computation of a linear model}}
  \begin{flushleft}
     \textsc{Input:}
     \begin{itemize}
        \setlength{\itemsep}{0pt}
        \setlength{\parskip}{0pt}
        \setlength{\parsep}{0pt}
    \item	A set $\Xall$ of observations in $\R^D$
     \item	The target dimension $d$ for the linear model,  where $d \in \{1, 2, \dots, D - 1\}$
     \end{itemize}
     \textsc{Output:}
     \begin{itemize}
        \setlength{\itemsep}{0pt}
        \setlength{\parskip}{0pt}
        \setlength{\parsep}{0pt}
    \item  A $d$-dimensional subspace of $\R^D$
     \end{itemize}
     \textsc{Procedure:}
     \renewcommand{\theenumi}{\footnotesize \arabic{enumi} \ }
	 \renewcommand{\labelenumi}{\theenumi}	
      \begin{enumerate} \setlength{\itemsep}{4pt}
        \setlength{\itemsep}{0pt}
        \setlength{\parskip}{0pt}
        \setlength{\parsep}{0pt}
      \item {}({\sl Optional.})
      	Solve~\eqref{eqn:euclid-median} to obtain a center $\vct{c}_{\star}$, and
      	center the data: $\vct{x} \leftarrow \vct{x} - \vct{c}_{\star}$ for each $\vct{x} \in \Xall$
      \item {}({\sl Optional.}) Spherize the data: $\vct{x} \leftarrow \vct{x}/\enorm{\vct{x}}$
      	for each nonzero $\vct{x} \in \Xall$
      \item Solve the \rrp\ problem~\eqref{eqn:rrp} with dataset $\Xall$ and parameter $d$ to obtain an optimal point $\mtx{P}_{\star}$ \label{proto:solve}
      \item Solve~\eqref{eqn:closest-proj} by finding a dominant $ d $-dimensional invariant subspace of $\mtx{P}_{\star}$
      \end{enumerate}
    \end{flushleft}
    \label{alg:proto}
\end{algorithm}

First, the \rrp\ problem assumes that the inliers are approximately centered.  When they are not, it is important to identify a centering point $\vct{c}_{\star}$ for the dataset and to work with the centered observations.
We can compute a centering point $\vct{c}_{\star}$ robustly by solving the Euclidean median problem~\cite{HR09:Robust-Statistics,MMY06:Robust-Statistics}:
\begin{equation} \label{eqn:euclid-median}
\minimize	\sum_{\vct{x} \in \Xall} \norm{ \vct{x} - \vct{c} }.
\end{equation}
It is also possible to incorporate centering by modifying the optimization problem~\eqref{eqn:rrp}.  For brevity, we omit the details.

Second, the \rrp\ formulation can be sensitive to outliers with large magnitude.  A simple but powerful method for addressing this challenge is to spherize the data points before solving the optimization problem.
For future reference,
we write down the resulting convex program.
\begin{equation} \label{eqn:srrp}
\minimize	\sum_{\vct{x} \in \Xall} \norm{ \tvct{x} - \mtx{P} \tvct{x} }
\subjto		\mtx{0} \psdle \mtx{P} \psdle \Id
\quad\textsf{and}\quad
\trace \mtx{P} = d.
\end{equation}
The tilde denotes the spherization transform~\eqref{eqn:tilde}.
We refer to~\eqref{eqn:srrp} as the \srrp\ problem.
In most (but not all) of our experimental work, we have found that \srrp\ outperforms \rrp.
The idea of spherizing data before fitting a subspace was proposed in the paper~\cite{LMS+99:Robust-Principal},
where it is called \term{spherical PCA}.

Finally, we regard the parameter $d$ in \rrp\ and \srrp\ as a proxy for the dimension of the linear model.
While the rank of an optimal solution $\mtx{P}_{\star}$ to~\eqref{eqn:rrp} or~\eqref{eqn:srrp}
cannot be smaller than $d$ because of the constraints
$\trace \mtx{P} = d$ and $\mtx{P} \psdle \Id$,
the rank of $\mtx{P}_{\star}$ often exceeds $d$.
We recommend solving~\eqref{eqn:closest-proj} to
find a closest projector to $\mtx{P}_{\star}$.

\subsection{Experimental Setup}

To solve the \rrp\ problem~\eqref{eqn:rrp} and the \srrp\ problem~\eqref{eqn:srrp},
we use the IRLS method, Algorithm~\ref{alg:IRLS}. %
We set the regularization parameter $\delta = 10^{-10}$ and the stopping tolerance
$\eps = 10^{-15}$.  We postprocess the computed optimal point $\mtx{P}_{\star}$
of \rrp\ or \srrp\ to obtain a $d$-dimensional linear model by solving~\eqref{eqn:closest-proj}.

\subsection{Comparisons}

By now, there are a huge number of proposals for robust linear modeling,
so we have limited our attention to methods that are computationally tractable.  That is,
we consider only formulations that have a polynomial-time algorithm for constructing a
global solution (up to some tolerance).  We do not discuss
techniques that involve Monte Carlo simulation, nonlinear programming, etc.\ %
because the success of these approaches depends largely on parameter settings and providence.  As a consequence, it is hard to evaluate their behavior in a consistent way.

We consider two standard approaches, PCA~\cite{Jol02:Principal-Component}
and spherical PCA~\cite{LMS+99:Robust-Principal}.  Spherical PCA rescales each observation so it lies on the Euclidean sphere, and then it applies standard PCA.
Simulations performed with several different robust PCA methods in~\cite{Maronna2005} lead Maronna et al.~\cite{MMY06:Robust-Statistics}
to recommend spherical PCA as a reliable classical robust PCA algorithm.

We also consider a more recent proposal~\cite{XCS10:Robust-PCA-NIPS,XCS11:Robust-PCA,MT11:Two-Proposals},
which is called \emph{low-leverage decomposition} (LLD) or \emph{outlier pursuit}.
This method decomposes the $D \times N$ matrix $\mtx{X}$ of observations
by solving the optimization problem
\begin{equation} \label{eqn:lld-exp}
\minimize \pnorm{S_1}{\mtx{P}} + \gamma \, \pnorm{1\to2}{\mtx{C}}^*
\subjto \mtx{X} = \mtx{P} + \mtx{C}
\end{equation}
where $\pnorm{S_1}{\cdot}$ is the Schatten 1-norm and $\pnorm{1\to2}{\cdot}^*$
returns the sum of Euclidean norms of the column.  The idea is that the
optimizer $(\mtx{P}_{\star}, \mtx{C}_{\star})$ will consist of a low-rank
model $\mtx{P}_{\star}$ for the data along with a column-sparse matrix $\mtx{C}_{\star}$
that identifies the outliers.  We always use the parameter choice
$\gamma = 0.8 \sqrt{D/N}$, which seems to be effective in practice.

We do not make comparisons with the rank--sparsity decomposition~\cite{CSPW09:Rank-Sparsity},
which has also been considered for robust linear modeling in~\cite{CLMW11:Robust-Principal}.
It is not effective for the problem that we consider here.

\subsection{Faces in a Crowd}
\label{sec:faces-crowd}

This experiment is designed to test how well several robust methods
are able to fit a linear model to face images that are dispersed in a
collection of random images.  Our setup allows us to study how well
the robust model generalizes to faces we have not seen.

We pull 64 images of a single face under different illuminations from the
Extended Yale Face Database B~\cite{LHK05:Acquiring-Linear}.
We use the first 32 faces
for the sample, and we reserve the other 32 for the out-of-sample test.
Next, we add all 467 images from the
\texttt{BACKGROUND/Google} folder of the Caltech101 database~\cite{CIT101,Fei-Fei2004}.
The Caltech101 images are converted to grayscale and downsampled to \(192\times 168\) pixels to match the native resolution of the Yale face images.
We center the images by subtracting the Euclidean median~\eqref{eqn:euclid-median}.
Then we apply PCA, spherical PCA, LLD, \rrp, and \srrp\ to fit a nine-dimensional subspace
to the data.
See~\cite{BJ03:Lambertian-Reflectance} for justification of the choice $d = 9$.
This experiment is similar to work reported in~\cite[Sec.~VI]{LLY+10:Robust-Recovery}.

Figure~\ref{fig:face-ims} displays several images
from the sample projected onto the computed nine-dimensional subspace
(with the centering added back after projection).  For every method, the projection
of an in-sample face image onto the subspace is recognizable as a face.
Meanwhile, the out-of-sample faces are described poorly by the PCA subspace.
All of the robust subspaces capture the facial features better, with \srrp\
producing the clearest images.

\begin{figure}[t!]
  \centering
  \includegraphics[width=\columnwidth]{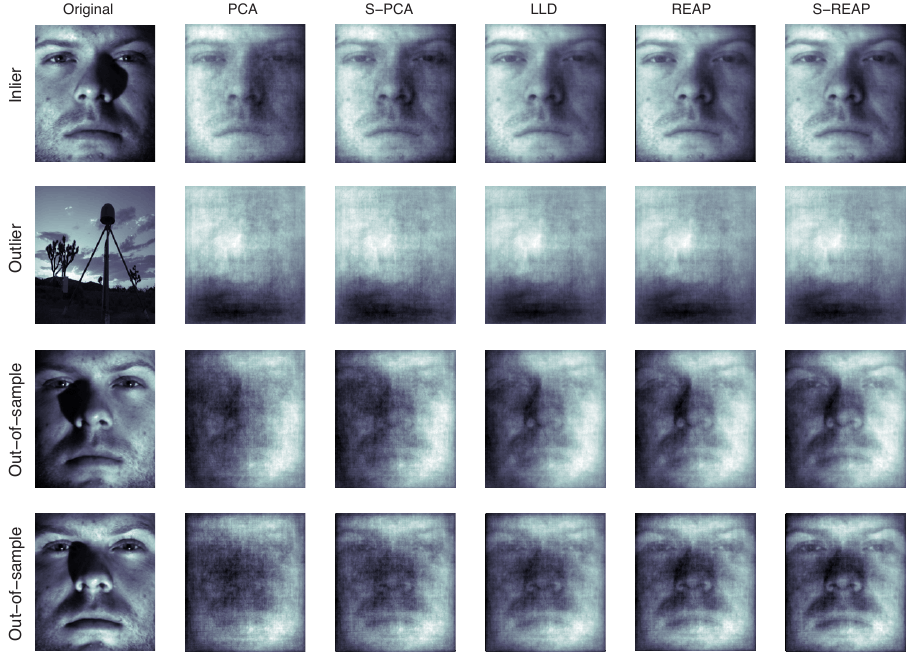}
  \caption{\textsl{Face images projected onto nine-dimensional linear model.}
  The dataset consists of 32 images of a single face under different illuminations
  and 400 random images from the Caltech101 database.
  The original images (\textsl{left column}) are projected onto the
  nine-dimensional subspaces computed using five different
  modeling techniques.  The first two rows indicate how well the models
  explain the in-sample faces versus the random images.  The last two rows
  show projections of two out-of-sample faces, which were not used to
  compute the linear models.  See Section~\ref{sec:faces-crowd}
  for details.}
  \label{fig:face-ims}
\end{figure}

Figure~\ref{fig:face-dist-graph} shows the ordered distances of the
32 out-of-sample faces to the robust linear model
as a function of the ordered distances to the model computed with PCA.
A point below the 1:1 line means that the $i$th closest point is closer
to the robust model than the $i$th closest point is to the PCA model.
Under this metric, \srrp\ is the dominant method, which explains the qualitative behavior
seen in Figure~\ref{fig:face-ims}.
This plot clearly demonstrates that \srrp\ computes a subspace that
generalizes better than the subspaces obtained with the other robust methods.

\begin{figure}[t!]
  \centering
  \includegraphics[width=\columnwidth]{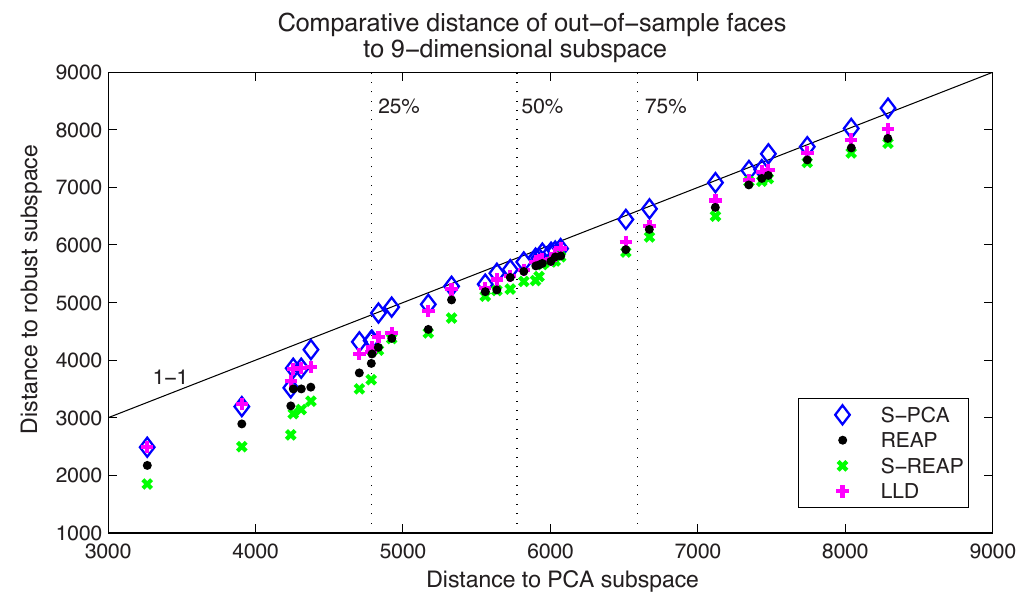}
  \caption{\textsl{Approximation of out-of-sample face images by several linear models.}
  The ordered distances of the out-of-sample face images to each robust model
  as a function of
  the ordered distances to the PCA model.
  The model was computed from 32 in-sample images;
  this graph shows how the model generalizes
  to the 32 out-of-sample images.
  Lower is better.
  See Section~\ref{sec:faces-crowd} for details.}
  \label{fig:face-dist-graph}
\end{figure}

\section{Related Work}
\label{sec:previous}

Robust linear modeling has been an active subject of research for over three decades.
Although many classical approaches have strong robustness properties \emph{in theory},
the proposals usually involve either intractable computational problems or  algorithms that are designed for a different problem like robust covariance estimation. More
recently, researchers have developed several techniques, based on convex optimization,
that are computationally efficient and admit some theoretical guarantees.
In this section, we summarize classical and contemporary work on robust
linear modeling, with a focus on the numerical aspects.
We recommend the books~\cite{HR09:Robust-Statistics,MMY06:Robust-Statistics,RL87:Robust-Regression}
for a comprehensive discussion of robust statistics.

\subsection{Classical Strategies for Robust Linear Modeling}

We begin with an overview of the major techniques that have been proposed
in the statistics literature. %
The theoretical contributions in this area
focus on breakdown points and influence functions of estimators.
Researchers tend to devote less attention to the computational challenges
inherent in these formulations. Moreover, the notion of the breakdown point for quantifying the robustness of estimators for vectors and matrices does not generalize to subspace estimation.
We recall that, roughly speaking, the breakdown point measures
the proportion of arbitrarily placed outliers an estimator can handle before giving an arbitrarily bad result. However, this idea does not directly extend to subspaces since the set of subspaces is compact. Following instead \cite{CLMW11:Robust-Principal,XCS11:Robust-PCA}, we quantify robustness of subspace estimation via exact recovery and  stability to noise.

\subsubsection{Robust Combination of Residuals}

Historically, one of the earliest approaches to linear regression is to minimize the sum of (nonorthogonal) residuals.
This is the principle of \emph{least absolute deviations} (LAD).  Early proponents of this idea include Galileo, Boscovich, Laplace, and Edgeworth.  See~\cite{Har74:Method-Least-I,Har74:Method-Least-II,Dod87:Introduction-L1-Norm} for some historical discussion.  It appears that \emph{orthogonal} regression with LAD was first
considered in the late 1980s~\cite{OW85:Analysis-Total,SW87:Orthogonal-Linear,Nyq88:Least-Orthogonal};
the extension from orthogonal regression to PCA seems to be even more
recent~\cite{Wat01:Some-Problems,DZHZ06:R1-PCA}.
LAD has also been considered as a method for hybrid linear modeling
in~\cite{ZSL09:Median-k-Flats,LZ11:Robust-Recovery}.
We are not aware of a tractable algorithm for these formulations.

There are many other robust methods for combining residuals aside from LAD.
An approach that has received wide attention is to minimize the median of the squared
residuals~\cite{Rou84:Least-Median,RL87:Robust-Regression}.  Other methods appear in
the books~\cite{HR09:Robust-Statistics,MMY06:Robust-Statistics}.  These formulations
are generally not computationally tractable.

\subsubsection{Robust Estimation of Covariance Matrix}

Another standard technique for robust PCA is to form a robust estimate of the covariance
matrix of the data~\cite{Mar76:Robust-M-Estimators,HR09:Robust-Statistics,MMY06:Robust-Statistics,DGK81:Robust-Estimation,Dav87:Asymptotic-Behavior,XY95:Robust-Principal,CH00:Principal-Component}.
The classical approaches to robust estimates of covariance are based on maximum likelihood principles that lead to M-estimators.  Most often, IRLS algorithms are used to compute these M-estimators of covariance.  

There are some formal similarities between the minimization program~\eqref{eqn:rrp} and the formulation of classical M-estimators. In particular,  the computational complexity of these estimators scales comparably with our own IRLS algorithm. However, there are also important differences between the classical covariance estimators and the subspace estimation procedure we consider here.  In particular, the common M-estimators for robust covariance estimation fail for the exact- and near-subspace recovery problems.   See~\cite[Sec.~3.1]{ZL11:Novel-M-Estimator} for elaboration on these points.

In addition to the basic M-estimators computed with IRLS, there are many other covariance estimators, including S-estimators, the minimum covariance determinant (MCD),  the minimum volume ellipsoid (MVE), and the Stahel--Donoho estimator.
We are not aware of any scalable algorithm with  guarantees of correctness for implementing these latter estimators.
See~\cite[Sec.~6]{MMY06:Robust-Statistics} for a review.

\subsubsection{Projection Pursuit PCA}

Projection pursuit (often abbreviated PP-PCA) is a procedure that constructs principal components
one at a time by finding a direction that maximizes a robust measure of scale,
removing the component of the data in this direction, and repeating the process.
The initial proposal appears in~\cite[1st edn.]{HR09:Robust-Statistics},
and it has been explored by many other
authors~\cite{LC85:Projection-Pursuit,Amm93:Robust-Singular,CFO07:Algorithms-Projection,Kwa08:Principal-Component,WTH09:Penalized-Matrix}.
We are aware of only one formulation that provably (approximately) maximizes a robust measure of scale
at each iteration~\cite{MT11:Two-Proposals}, but there are no overall
guarantees for PP-PCA algorithms.

\subsubsection{Screening for Outliers}

Another common approach is to remove possible outliers and then estimate the underlying subspace by PCA~\cite{CW82:Residuals-Influence,TB01:Robust-Principal,TB03:Framework-Robust}.
The classical methods offer very limited guarantees.  There are some recent algorithms that are provably
correct~\cite{Bru09:Robust-PCA,XCM10:Principal-Component} under some model assumptions and with particular correctness
criteria that are tailored to the individual algorithm.

\subsubsection{RANSAC}

The randomized sample consensus (RANSAC) method
is a randomized iterative procedure for fitting models to noisy data
consisting of inliers and outliers~\cite{FB81:Random-Sample}.
Under some assumptions, RANSAC will eventually identify
a linear model for the inliers, but there are no guarantees on the number of iterations required.

\subsubsection{Spherical PCA}

A useful method for fitting a robust linear model is to
center the data robustly, project it onto a sphere, and then apply standard PCA.
This approach is due to~\cite{LMS+99:Robust-Principal}.
Maronna et al.~\cite{MMY06:Robust-Statistics}
recommend it as a preferred method for robust PCA.  The technique is
computationally practical, but it has limited theoretical guarantees.

\subsection{Approaches Based on Convex Optimization}

Recently, researchers have started to develop effective techniques for
robust linear modeling that are based on convex optimization.  These
formulations invite a variety of tractable algorithms,
and they have theoretical guarantees under appropriate model assumptions.

\subsubsection{Demixing Methods}

One class of techniques for robust linear modeling is based on splitting a data
matrix into a low-rank model plus a corruption.  The first approach of this form is
due to Chandrasekaran et al.~\cite{CSPW09:Rank-Sparsity}.
Given an observed matrix $\mtx{X}$, they solve the semidefinite problem
\begin{equation} \label{eqn:rank-sparsity}
\minimize	\pnorm{S_1}{\mtx{P}} + \gamma \pnorm{\ell_1}{\mtx{C}}
\subjto \mtx{X} = \mtx{P} + \mtx{C}.
\end{equation}
Minimizing the Schatten 1-norm $\pnorm{S_1}{\cdot}$ promotes low rank, while
minimizing the vector $\ell_1$ norm promotes sparsity.  The regularization
parameter $\gamma$ negotiates a tradeoff between the two goals.
Cand{\`e}s et al.~\cite{CLMW11:Robust-Principal} study the performance of~\eqref{eqn:rank-sparsity}
for robust linear modeling in the setting where individual entries of
the matrix $\mtx{X}$ are subject to error.

A related proposal is due to Xu et al.~\cite{XCS10:Robust-PCA-NIPS,XCS11:Robust-PCA}
and independently to McCoy \& Tropp~\cite{MT11:Two-Proposals}.
These authors recommend solving the decomposition problem
\begin{equation} \label{eqn:lld}
\minimize	\pnorm{S_1}{\mtx{P}} + \gamma \pnorm{1\to2}{\mtx{C}}^*
\subjto \mtx{X} = \mtx{P} + \mtx{C}.
\end{equation}
The norm $\pnorm{1 \to 2}{\cdot}^{*}$ returns the sum of Euclidean norms of the columns of
its argument.  This formulation is appropriate for inlier--outlier data models, where
entire columns of the data matrix may be corrupted, in contrast to the formulation~\eqref{eqn:rank-sparsity} that is used for corruptions of individual
matrix elements.

Both \eqref{eqn:rank-sparsity} and~\eqref{eqn:lld} possess some theoretical guarantees under appropriate model assumptions, but we restrict our discussion to~\eqref{eqn:lld} because it is tuned to the In \& Out Model that we consider here. In the noiseless case, Xu et al.~\cite{XCS11:Robust-PCA} show that~\eqref{eqn:lld} will exactly recover the underlying subspace under the In \& Out Model so long
as the inlier-to-outlier ratio exceeds a constant times  the inlier dimension $d$.\footnote{More precisely, the inlier-to-outlier ratio must exceed $(121\mu/9) d $, where $\mu\ge 1$ depends on the data.} For the Haystack Model with $\sigma_{\In}=\sigma_{\Out}$ and $d \ll D$, the lower bound~\eqref{eq:stab-bd-simple} is positive when the fraction of inliers exceeds a constant times $d/D$.  Hence, Theorem~\ref{thm:rrp-stable} endows \rrp\ with an  exact recovery guarantee  that is stronger than the results of~\cite{XCS11:Robust-PCA} for~\eqref{eqn:lld} in the \(d\ll D\) regime; a similar statement holds for the stability of \rrp.  Moreover,  the work of Coudron \& Lerman~\cite{CL12:Sample-Complexity-RPCA}, which appeared after the submission of this work,  provides sample-complexity guarantees for \rrp\ that mirrors that of standard PCA when the data \(\Xall\) is drawn i.i.d. from a subgaussian distribution.

The most common algorithmic framework for demixing methods of the form~\eqref{eqn:rank-sparsity} and~\eqref{eqn:lld} uses the alternating direction method of multipliers (ADMM)~\cite{BPC10:Distributed-Optimization}.
These algorithms can converge slowly, so it may take excessive computation to obtain a high-accuracy solution. Indeed, our limited numerical experiments indicate that \rrp\ tends to be significantly faster than the ADMM implementation of~\cite{XCS10:Robust-PCA-NIPS,XCS11:Robust-PCA}
and~\cite{MT11:Two-Proposals}.  

Nevertheless, the demixing strategy readily adapts to different situations such as missing  observations or entrywise corruptions~\cite{XCS11:Robust-PCA,CLMW11:Robust-Principal}, while it is not immediately clear how to adapt the \rrp\ framework to these scenarios.

\subsubsection{Precedents for the \rrp\ Problem}

The \rrp\ problem~\eqref{eqn:rrp} is a semidefinite relaxation of the $\ell_1$ orthogonal distance  problem~\eqref{eqn:orlav}.  Our work extends an earlier %
relaxation of~\eqref{eqn:orlav} proposed by Zhang \& Lerman~\cite{ZL11:Novel-M-Estimator}: %
\begin{equation}
\minimize \sum_{\vct{x} \in \Xall} \norm{ \vct{x} - \mtx{P} \vct{x} }
\subjto \trace( \Id -\mtx{P}) = 1,\label{eq:GMS}
\end{equation}
where the minimum occurs over symmetric \(\mtx P\). Although not obvious, the formulation above is equivalent to \rrp\ with the specific choice \(d=D-1\).  Indeed, any optimal point \(\mtx P_\star\)  of~\eqref{eq:GMS} satisfies \(\Id -\mtx P_\star \psdge \zerovct\)~\cite[Lem.~14]{ZL11:Novel-M-Estimator}, and this fact, together with the  trace constraint \(\trace (\Id -\mtx P_\star) = 1\), implies that \(\Id- \mtx P_\star \psdle \Id\).  Thus, the \rrp\ constraints \(\zerovct \psdle \mtx P\psdle\Id\) are implicit in~\eqref{eq:GMS}, and so the observation \(\trace(\Id) = D\) yields the claimed equivalence.

 The present work extends the earlier formulation by freeing the parameter \(d\)  to search for subspaces of a specific dimension, which provides a tighter relaxation for finding \(d\)-dimensional orthoprojectors. In~\cite{ZL11:Novel-M-Estimator}, however, the authors show that the optimal point \(\mtx P_\star\) of~\eqref{eq:GMS} is more analogous to a robust inverse covariance matrix than to an approximate orthoprojector.  This allows the determination of the dimension \(d\) using the eigenvalues of  \(\mtx P_\star\), while in the present work, we  treat \(d\) as a known parameter. 

Our  analysis of~\rrp\ builds on the ideas first presented in~\cite{LZ10:lp-Recovery,ZL11:Novel-M-Estimator}, but it incorporates a number of refinements that simplify and improve the theoretical guarantees.  In particular, the present results do not require  an oracle condition like~\cite[Eqs.~(8,9)]{ZL11:Novel-M-Estimator}, and our stability statistic \(\Stab(L)\) supersedes the earlier exact recovery and stability requirements~\cite[Eqs.~(6,7) \& (10,11)]{ZL11:Novel-M-Estimator}. The  exact recovery guarantees  under the  Haystack Model are somewhat stronger for \rrp\  than the analogous guarantees for~\eqref{eq:GMS} \cite[Sec.~2.6.1]{ZL11:Novel-M-Estimator}. %
The IRLS algorithm for~\rrp\ and the convergence
analysis that we present in Section~\ref{sec:algorithm} also extend ideas from the earlier work.

From a broad perspective, the idea of relaxing a difficult nonconvex program like~\eqref{eqn:orlav} to obtain a convex problem is well
established in the literature on combinatorial optimization.  Research on
linear programming relaxations is summarized in~\cite{Vaz03:Approximation-Algorithms}.
Some significant works on semidefinite relaxation
include~\cite{LS91:Cones-Matrices,GW95:Improved-Approximation}.

\appendix
\section{Theorem~\ref{thm:rrp-stable}: Supporting lemmas} \label{sec:pf-rrp-lemmas}

This appendix contains the proofs of the two technical results that animate
Theorem~\ref{thm:rrp-stable}.  Throughout, we maintain
the notation and assumptions from the In \& Out Model on page~\pageref{tab:in-out}
and from the statement and proof of Theorem~\ref{thm:rrp-stable}.

\subsection{Controlling the Size of the Perturbation} \label{sec:control-perturb}

In this section, we establish Lemma~\ref{lem:control-perturb}.  On comparing the objective function $f$ from~\eqref{eqn:rrp-objective} with the perturbed objective function $g$ from~\eqref{eqn:perturbed-objective}, we see that the only changes involve the terms containing inliers.  As a consequence, the difference function $h = f - g$ depends only on the inliers:
$$
h(\mtx{P}) = \sum_{\vct{x} \in \Xin} \big[ \enorm{ (\Id - \mtx{P}) \vct{x} }
	- \enorm{ (\Id - \mtx{P}) \Proj_{L} \vct{x} } \big],
$$
where $\mtx{P}$ is an arbitrary  matrix.
Apply the lower triangle inequality to see that
\begin{align} \label{eqn:perturbation-bd}
\abs{ h(\mtx{P}) }
	&\leq \sum_{\vct{x} \in \Xin} \enorm{ (\Id - \mtx{P})(\Id - \Proj_{L}) \vct{x} } \notag \\
	&\leq \sum_{\vct{x} \in \Xin} \big[ \enorm{ (\Id - \Proj_{L}) \vct{x} }
	+ \enorm{ (\Proj_{L} - \mtx{P})(\Id - \Proj_{L}) \vct{x} } \big]
	\notag \\
	&\leq \big[ 1 + \norm{ \Proj_{L} - \mtx{P} } \big] \cdot
	 \sum_{\vct{x} \in \Xin} \enorm{ (\Id - \Proj_{L}) \vct{x} }
	\notag \\
	&\leq \big[ 1 + \pnorm{S_1}{ \Proj_{L} - \mtx{P} } \big] \cdot \Resid(L).
\end{align}
To justify the second inequality, write $\Id - \mtx{P} = (\Id - \Proj_{L}) + (\Proj_{L} - \mtx{P})$, and then apply the upper triangle inequality.  Use the fact that the projector $\Id - \Proj_{L}$ is idempotent to simplify the resulting expression.  The third inequality depends on the usual operator-norm bound.  We reach~\eqref{eqn:perturbation-bd} by identifying the sum as the total inlier residual $\Resid(L)$, defined in~\eqref{eqn:residual}, and invoking the fact that the Schatten 1-norm dominates the operator norm.

Applying~\eqref{eqn:perturbation-bd} twice, it becomes apparent that
$$
\abs{ h(\Proj_{L}) - h( \Proj_{L} + \mtx{\Delta} ) }
	\leq \abs{ h(\Proj_{L}) } + \abs{ h( \Proj_{L} + \mtx{\Delta} ) }
	\leq \big[ 2 + \pnorm{S_1}{\mtx{\Delta}} \big] \, \Resid(L).
$$
There are no restrictions on the matrix $\mtx{\Delta}$, so the demonstration of Lemma~\ref{lem:control-perturb} is complete.

\subsection{The Rate of Ascent of the Perturbed Objective} \label{sec:rate-ascent}

In this section, we establish Lemma~\ref{lem:rate-ascent}.  This result contains the essential insights behind Theorem~\ref{thm:rrp-stable}, and the proof involves some amount of exertion.

Assume that $\Proj_{L} + \mtx{\Delta} \in \Phi$.  Since the perturbed objective $g$ defined in~\eqref{eqn:perturbed-objective} is a continuous convex function, we have the lower bound
\begin{equation} \label{eqn:rate-ascent-bd}
g(\Proj_{L} + \mtx{\Delta}) - g( \Proj_{L} )
	\geq g'(\Proj_{L}; \mtx{\Delta}).
\end{equation}
The right-hand side of~\eqref{eqn:rate-ascent-bd} refers to the one-sided directional derivative of $g$ at the point $\Proj_{L}$ in the direction $\mtx{\Delta}$.  That is,
$$
g'(\Proj_{L}; \mtx{\Delta})
	= \lim_{t \downarrow 0} \frac{1}{t} \big[ g(\Proj_{L} + t \mtx{\Delta}) - g(\Proj_{L}) \big].
$$
We can be confident that this limit exists and takes a finite value~\cite[Thm.~23.1 et seq.]{Roc70:Convex-Analysis}.

Let us find a practicable expression for the directional derivative.  Recalling the definition~\eqref{eqn:perturbed-objective} of the perturbed objective, we see that the difference quotient takes the form
\begin{multline} \label{eqn:g-diff-quot}
\frac{1}{t} \big[ g(\Proj_L + t \mtx{\Delta}) - g(\Proj_L) \big] \\
	 = \sum_{\vct{x} \in \Xin} \frac{1}{t} \big[\enorm{(\Proj_{L^\perp} - t\mtx{\Delta}) \Proj_L \vct{x}} - \enorm{ \Proj_{L^\perp} \Proj_L \vct{x} } \big]
	+ \sum_{\vct{x} \in \Xout} \frac{1}{t} \big[ \enorm{(\Proj_{L^\perp} - t \mtx{\Delta}) \vct{x} }
	- \enorm{ \Proj_{L^\perp} \vct{x} } \big].
\end{multline}
We analyze the two sums separately.  First, consider the terms that involve inliers.
For $\vct{x} \in \Xin$,
$$
\frac{1}{t} \big[ \enorm{ (\Proj_{L^\perp} - t \mtx{\Delta}) \Proj_{L} \vct{x} } - 	\enorm{ \Proj_{L^\perp} \Proj_{L} \vct{x} } \big] \\
	= \enorm{\mtx{\Delta} \Proj_{L} \vct{x} }
	\quad\text{for all $t > 0$.}
$$
We have used the fact that $\Proj_{L^\perp} \Proj_L = \mtx{0}$ twice.
Next, consider the terms involving outliers.  By the assumptions of the In \& Out Model on page~\pageref{tab:in-out}, each outlier $\vct{x} \in \Xout$ has a nontrivial component in the subspace $L^{\perp}$.  We may calculate that
\begin{align*}
\enorm{ (\Proj_{L^\perp} - t \mtx{\Delta}) \vct{x} }
	&= \left[ \vsmnorm{}{\Proj_{L^\perp} \vct{x}}^2
		- 2 \ip{ t \mtx{\Delta} \vct{x} }{ \Proj_{L^\perp} \vct{x} } + \enormsq{t \mtx{\Delta}\vct{x}} \right]^{1/2} \\
	&= \vsmnorm{}{\Proj_{L^\perp} \vct{x}} \left[ 1
		- \frac{2t}{\vsmnorm{}{\Proj_{L^\perp} \vct{x}}} \ip{ \mtx{\Delta} \vct{x} }{ \widetilde{\Proj_{L^\perp} \vct{x}}}
		+ \bigOh \big( t^2 \big) \right]^{1/2} \\
	&= \vsmnorm{}{\Proj_{L^\perp} \vct{x}} - t \ip{ \mtx{\Delta} \vct{x} }{ \widetilde{\Proj_{L^\perp} \vct{x}} }
		+ \bigOh \big( t^2 \big)
		\quad\text{as $t \downarrow 0$,}
\end{align*}
where the spherization transform \(\vct x \mapsto \tvct x\) is defined in~\eqref{eqn:tilde}. Therefore,
$$
\frac{1}{t} \big[ \enorm{ (\Proj_{L^\perp} - t \mtx{\Delta}) \vct{x} }
	- \vsmnorm{}{\Proj_{L^\perp} \vct{x}} \big]
	\to - \ip{ \mtx{\Delta} \vct{x} }{ \widetilde{\Proj_{L^\perp} \vct{x}} }
	\quad\text{as $t \downarrow 0$.}
$$
Introducing these facts into the difference quotient~\eqref{eqn:g-diff-quot} and taking the limit as $t \downarrow 0$, we determine that
\begin{equation} \label{eqn:directional-derivative}
g'(\Proj_{L}; \mtx{\Delta})
	= \sum_{\vct{x} \in \Xin} \enorm{ \mtx{\Delta} \Proj_{L} \vct{x} }
	- \sum_{\vct{x} \in \Xout} \ip{ \mtx{\Delta} \vct{x} }{ \widetilde{\Proj_{L^\perp} \vct{x}}}.
\end{equation}
It remains to produce a lower bound on this directional derivative.

The proof of the lower bound has two components.  We require a bound on the sum over outliers, and we require a bound on the sum over inliers.  These results appear in the following sublemmas, which we prove in Sections~\ref{sec:outliers-slem} and~\ref{sec:inliers-slem}.

\begin{sublemma}[Outliers] \label{slem:outliers}
Under the prevailing assumptions,
$$
\sum_{\vct{x} \in \Xout}
	\ip{ \mtx{\Delta} \vct{x} }{ \widetilde{\Proj_{L^\perp} \vct{x}} }
	\leq \vsmnorm{}{ \widetilde{\Proj_{L^\perp} \mtx{X}_{\Out}} }
	\cdot \norm{ \mtx{X}_{\Out} } \cdot
	\pnorm{S_1}{ \mtx{\Delta} }
	= \Align(L) \cdot \pnorm{S_1}{ \mtx{\Delta} },
$$
where the alignment statistic $\Align(L)$ is defined in~\eqref{eqn:structure}.
\end{sublemma}

\begin{sublemma}[Inliers] \label{slem:inliers}
Under the prevailing assumptions,
$$
\sum_{\vct{x} \in \Xin} \enorm{ \mtx{\Delta} \Proj_{L} \vct{x} }
	\geq \bigg[ \frac{1}{4\sqrt{d}} \cdot \inf_{\substack{\vct{u} \in L \\ \enorm{\vct{u}} = 1}}
	\ \sum_{\vct{x} \in \Xin} \absip{ \vct{u} }{ \vct{x} } \bigg] \cdot \pnorm{S_1}{\mtx{\Delta}}
	= \frac{\Perm(L)}{4\sqrt{d}}\cdot \pnorm{S_1}{\mtx{\Delta}},
$$
where the permeance statistic $\Perm(L)$ is defined in~\eqref{eqn:permeance}.
\end{sublemma}

To complete the proof of Lemma~\ref{lem:rate-ascent}, we substitute the bounds from Sublemmas~\ref{slem:outliers} and~\ref{slem:inliers} into the expression~\eqref{eqn:directional-derivative} for the directional derivative.  This step yields
$$
g'(\Proj_{L}; \mtx{\Delta})
	\geq \bigg[ \frac{\Perm(L)}{4\sqrt{d}} - \Align(L) \bigg] \cdot \pnorm{S_1}{\mtx{\Delta}}.
$$
Since $d = \dim(L)$, we identify the bracket as the stability statistic $\coll{S}(L)$ defined in~\eqref{eqn:stability}.  Combine this bound with~\eqref{eqn:rate-ascent-bd} to finish the argument.

\subsubsection{Outliers} \label{sec:outliers-slem}

The result in Sublemma~\ref{slem:outliers} is straightforward.  First, observe that
$$
\sum_{\vct{x} \in \Xout} \ip{ \mtx{\Delta} \vct{x} }{ \widetilde{\Proj_{L^\perp} \vct{x}}}
	= \sum_{\vct{x} \in \Xout} \ip{ \mtx{\Delta} }{ \widetilde{\Proj_{L^\perp} \vct{x}} \cdot \vct{x}^\transp }
	= \ip{ \mtx{\Delta} }{ \widetilde{\Proj_{L^\perp} \mtx{X}_{\Out}} \cdot \mtx{X}_{\Out}^\transp }.
$$
The first relation follows from the fact that $\ip{\mtx{A}\vct{b}}{\vct{c}} = \ip{ \mtx{A} }{\vct{c}\vct{b}^\transp}$.  We obtain the second identity by drawing the sum into the inner product and recognizing the product of the two matrices.  Therefore,
$$
\abs{ \sum_{\vct{x} \in \Xout} \ip{ \mtx{\Delta} \vct{x} }{ \widetilde{\Proj_{L^\perp} \vct{x}}} }
	\leq \vsmnorm{}{ \widetilde{\Proj_{L^\perp} \mtx{X}_{\Out}} \cdot \mtx{X}_{\Out}^\transp } \cdot \pnorm{S_1}{\mtx{\Delta}}.
$$
The last bound is an immediate consequence of the duality between the spectral norm and the Schatten 1-norm.  To complete the argument, we apply the operator norm bound $\norm{\mtx{AB}^\transp} \leq \norm{\mtx{A}} \norm{\mtx{B}}$.

\subsubsection{Inliers} \label{sec:inliers-slem}

The proof of Sublemma~\ref{slem:inliers} involves several considerations.  We first explain the overall structure of the argument and then verify that each part is correct.
The first ingredient is a lower bound for the minimum of the sum over inliers:
\begin{equation} \label{eqn:inlier-fnorm-min}
\sum_{\vct{x} \in \Xin} \enorm{ \mtx{\Delta} \Proj_{L} \vct{x} }
	\geq \left[ \inf_{\substack{\vct{u} \in L \\ \enorm{\vct{u}} = 1}}
	\ \sum_{\vct{x} \in \Xin} \absip{ \vct{u} }{ \vct{x} } \right] \cdot \fnorm{\mtx{\Delta} \Proj_{L}}.
\end{equation}
See Section~\ref{sec:minim-with-frob}. The second step depends on a simple comparison between the Frobenius norm and the Schatten 1-norm of a matrix:
\begin{equation} \label{eqn:fnorm-snorm}
\fnorm{\mtx{\Delta} \Proj_{L}}
\geq \frac{1}{\sqrt{d}}	\pnorm{S_1}{\mtx{\Delta} \Proj_{L}}.
\end{equation}
 The inequality~\eqref{eqn:fnorm-snorm} follows immediately when we express the two norms in terms of singular values and apply the fact that $\mtx{\Delta} \Proj_{L}$ has rank $d$; it requires no further comment.  Third, we argue that
\begin{equation} \label{eqn:snorm-feasible}
\pnorm{S_1}{\mtx{\Delta} \Proj_{L}}
	\geq \frac{1}{4} \pnorm{S_1}{ \mtx{\Delta} }
	\quad\text{when $\Proj_{L} + \mtx{\Delta} \in \Phi$.}
\end{equation}
This demonstration appears in Section~\ref{sec:feasible-directions}. Combining the bounds~\eqref{eqn:inlier-fnorm-min},~\eqref{eqn:fnorm-snorm}, and~\eqref{eqn:snorm-feasible}, we obtain the result stated in Sublemma~\ref{slem:inliers}

\subsubsection{Minimization with a Frobenius-norm constraint}
\label{sec:minim-with-frob}
To establish~\eqref{eqn:inlier-fnorm-min}, we assume that $\fnorm{ \mtx{\Delta} \Proj_{L} } = 1$.  The general case follows from homogeneity.

Let us introduce a (thin) singular value decomposition $\mtx{\Delta}\Proj_{L} = \mtx{U\Sigma V}^\transp$.  Observe that each column $\vct{v}_1, \dots, \vct{v}_d$ of the matrix $\mtx{V}$ is contained in $L$.  In addition, the singular values $\sigma_j$ satisfy $\sum_{j=1}^d \sigma_j^2 = 1$ because of the normalization of the matrix.

We can express the quantity of interest as
$$
\sum_{\vct{x} \in \Xin} \enorm{ \mtx{\Delta} \Proj_{L} \vct{x} }
	= \sum_{\vct{x} \in \Xin} \enorm{ \mtx{\Sigma V}^\transp \vct{x} }
	= \sum_{\vct{x} \in \Xin} \bigg[ \sum_{j=1}^d \sigma_j^2 \, \abssqip{ \vct{v}_j }{ \vct{x} } \bigg]^{1/2}
$$
The first identity follows from unitary invariance of the Euclidean norm.  In the second relation, we have just written out the Euclidean norm in detail.  To facilitate the next step, abbreviate $p_j = \sigma_j^2$ for each index $j$.  Calculate that
\begin{align*}
\sum_{\vct{x} \in \Xin} \enorm{ \mtx{\Delta} \Proj_{L} \vct{x} }
	&= \sum_{\vct{x} \in \Xin} \bigg[ \sum_{j=1}^d p_j \, \abssqip{ \vct{v}_j }{ \vct{x} } \bigg]^{1/2} \\
	&\geq \min_{j} \sum_{\vct{x} \in \Xin} \absip{ \vct{v}_j }{ \vct{x} } \\
	&\geq \inf_{\substack{\vct{u} \in L \\ \enorm{\vct{u}} = 1}} \
	 \sum_{\vct{x} \in \Xin} \absip{ \vct{u} }{ \vct{x} }.
\end{align*}
Indeed, the right-hand side of the first line is a concave function of the variables $p_j$.  By construction, these variables lie in the convex set determined by the constraints $p_j \geq 0$ and $\sum_j p_j = 1$.  The minimizer of a concave function over a convex set occurs at an extreme point, which delivers the first inequality.  To reach the last inequality, recall that each $\vct{v}_j$ is a unit vector in $L$.  This expression implies~\eqref{eqn:inlier-fnorm-min}.

\subsubsection{Feasible directions}
\label{sec:feasible-directions}

Finally, we need to verify the relation~\eqref{eqn:snorm-feasible}, which states that $\pnorm{S_1}{\mtx{\Delta} \Proj_{L}}$ is comparable with $\pnorm{S_1}{\mtx{\Delta}}$ provided that $\Proj_{L} + \mtx{\Delta} \in \Phi$.  To that end, we decompose the matrix $\mtx{\Delta}$ into blocks:
\begin{equation}\label{eqn:delta-block-decomp}
\mtx{\Delta} =
    \underbrace{\Proj_{L} \mtx\Delta \Proj_{L} }_{=: \mtx{\Delta}_1}
    \ + \
    \underbrace{\Proj_{L^\perp} \mtx\Delta \Proj_{L} }_{=: \mtx{\Delta}_2}
    \ + \
    \underbrace{\Proj_{L} \mtx\Delta \Proj_{L^\perp}
    }_{=\phantom{:} \mtx{\Delta}_2^\transp}
    \ + \
    \underbrace{\Proj_{L^\perp} \mtx\Delta \Proj_{L^\perp}}_{=: \mtx{\Delta}_3}.
\end{equation}
We claim that
\begin{equation} \label{eqn:block-claim}
\pnorm{S_1}{ \mtx{\Delta}_1 } = \pnorm{S_1} { \mtx{\Delta}_3 }
\quad\text{whenever $\Proj_{L} + \mtx{\Delta} \in \Phi$.}
\end{equation}
Granted this identity, we can establish the equivalence of norms promised in~\eqref{eqn:snorm-feasible}.  Indeed,
\begin{align*}
\pnorm{S_1}{ \mtx{\Delta} \Proj_{L} }
	&\geq \max\big\{ \pnorm{S_1}{ \mtx{\Delta}_1 }, \ \pnorm{S_1}{ \mtx{\Delta}_2 } \big\} \\
	&\geq \frac{1}{4} \big[ 2 \pnorm{S_1}{ \mtx{\Delta}_1 } + 2 \pnorm{S_1}{ \mtx{\Delta}_2 } \big] \\
	&= \frac{1}{4} \big[ \pnorm{S_1}{\mtx{\Delta}_1} + \pnorm{S_1}{\mtx{\Delta}_3} +
	\pnorm{S_1}{\mtx{\Delta}_2} + \pnorm{S_1}{\mtx{\Delta}_2^\transp} \big] \\
	&\geq \frac{1}{4} \pnorm{S_1}{ \mtx{\Delta} }.
\end{align*}
The first inequality holds because projection reduces the Schatten 1-norm of a matrix.  The second inequality is numerical.  The equality in the third line depends on the claim~\eqref{eqn:block-claim}.  The last bound follows from the subadditivity of the norm and~\eqref{eqn:delta-block-decomp}.

To check~\eqref{eqn:block-claim}, we recall the definition of the feasible set:
$$
\Phi = \{ \mtx{P} : \mtx{0} \psdle \mtx{P} \psdle \Id
\quad\text{and}\quad \trace{\mtx{P}} = d\}.
$$
The condition $\Proj_L + \mtx{\Delta} \in \Phi$ implies the semidefinite relation $\Proj_{L} + \mtx{\Delta} \psdle \Id$.  Conjugating by the orthoprojector $\Proj_{L}$, we see that
$$
\Proj_{L} + \mtx{\Delta}_1 = \Proj_{L} (\Proj_{L} + \mtx{\Delta}) \Proj_{L} \psdle \Proj_{L}.
$$
As a consequence, $\mtx{\Delta}_1 \psdle \mtx{0}$.  Similarly,
the relation $\mtx{0} \psdle \Proj_L + \mtx{\Delta}$ yields
$$
\mtx{0} \psdle \Proj_{L^\perp} (\Proj_{L} + \mtx{\Delta}) \Proj_{L^\perp}
	= \mtx{\Delta}_3.
$$
Therefore, $\mtx{\Delta}_3 \psdge \mtx{0}$.
Since $\trace(\Proj_L + \mtx{\Delta}) = d$ and $\trace \Proj_{L} = d$, it is clear that $0 = \trace \mtx{\Delta} = \trace(\mtx{\Delta}_1 + \mtx{\Delta}_3)$.  We conclude that
$$
\pnorm{S_1}{\mtx{\Delta}_1} = \trace( - \mtx{\Delta}_1 )
	= \trace( \mtx{\Delta}_3 )
	= \pnorm{S_1}{\mtx{\Delta}_3}.
$$
The first and third equality hold because the Schatten 1-norm of a positive-semidefinite matrix coincides with its trace.

\section{Analysis of the Haystack Model}
\label{sec:proof-main-bound}

In this appendix, we establish exact recovery conditions for the Haystack Model.
To accomplish this goal, we study the probabilistic behavior of the
permeance statistic and the alignment statistic.
Our main result for the Haystack Model, Theorem~\ref{thm:gaussian-outliers},
follows when we introduce the probability bounds into the deterministic recovery
result, Theorem~\ref{thm:rrp-stable}.  The simplified result for the Haystack Model, Theorem~\ref{thm:random-rrp}, is a consequence of the following more detailed theory.

\begin{thm}\label{thm:gaussian-outliers}
Fix a parameter $c > 0$.
Let $L$ be an arbitrary $d$-dimensional subspace of $\R^D$, and draw the dataset
$\Xall$ at random according to the Haystack Model on page~\pageref{tab:haystack}.
Let $1 \leq d \leq D-1$.  The stability statistic satisfies
\begin{equation}\label{eqn:rrp-guarantee}
\Stab(L) \geq
\frac{\sigma_{\In}}{\sqrt{8 \pi}} \left[ \rho_{\In} - \frac{(2+c)}{\sqrt{2/\pi}} \sqrt{\rho_{\In}} \right] -
	\sigma_{\Out}\sqrt{\frac{D}{D - d -0.5}}
\left[ \sqrt{\rho_{\Out}} + 1 + c \sqrt{\frac{d}{D}} \right]^2
\end{equation}%
except with probability $3.5 \, \econst^{-c^2 d /2}$.
\end{thm}

We verify this expression below
in Sections~\ref{sec:proof-general-thm} and~\ref{sec:prf-gauss-outliers}.
Now, we demonstrate that Theorem~\ref{thm:gaussian-outliers} contains the simplified
result for the Haystack model, Theorem~\ref{thm:random-rrp}.

\begin{proof}[Proof of Theorem~\ref{thm:random-rrp} from Theorem~\ref{thm:gaussian-outliers}]
To begin, we collect some numerical inequalities.  For $\alpha>0$, the
function $f(x) = x - \alpha \sqrt{x}$ is convex when $x \geq 0$, so that
\begin{equation*}
x - \alpha \sqrt{x} = f(x) \geq f(\alpha^2) + f'(\alpha^2)(x-\alpha^2) = \frac{1}{2} (x - \alpha^2).
\end{equation*}
For $1 \leq d \leq (D-1)/2$, we have the numerical bounds
\begin{equation*}
\frac{D}{D-d-0.5} \leq 2 %
\quad\text{and}\quad
\frac{d}{D}\leq 0.5.
\end{equation*}
Finally, recall that $(a+b)^2 \leq 2(a^2 + b^2)$ and $(a+b+e)^2 \leq 3(a^2+b^2+e^2)$ as a consequence of H\"older's inequality.

To prove Theorem~\ref{thm:random-rrp},
we apply our numerical inequalities to weaken the bound~\eqref{eqn:rrp-guarantee}
from Theorem~\ref{thm:gaussian-outliers} to
\begin{equation*}
\Stab(L) \geq \frac{\sigma_\In}{\sqrt{32\pi}} \left[ \rho_\In - \pi(4+c^2) \right]
	- 6\sigma_\Out \left[ \rho_{\Out} + 1 + 0.5 c^2 \right]
\end{equation*}
Set $c = \sqrt{2 \beta}$ to reach the conclusion.
\end{proof}

\subsection{Tools for Computing the Summary Statistics}
\label{sec:proof-general-thm}

The proof of Theorem~\ref{thm:gaussian-outliers} requires probability
bounds on the permeance statistic $\Perm$ and
the alignment statistic $\Align$.
These bounds follow from tail inequalities for
Gaussian and spherically distributed random vectors
that we develop in the next two subsections.

\subsubsection{Tools for the Permeance Statistic}
\label{sec:key-condition}

In this section, we develop the probability inequality that
we need to estimate the permeance statistic $\Perm(L)$ for
data drawn from the Haystack model.

\begin{lemma}\label{lem:lower-bd-gauss}
Suppose $\vct{g}_1, \dots, \vct{g}_n$
are i.i.d.~$\normal(\vct{0}, \Id)$ vectors in $\R^d$.  For each $t \geq 0$,
\begin{equation} \label{eqn:gauss-lower-bd}
\inf_{\enorm{\vct{u}}=1} \ \sum_{i = 1}^{n} \absip{\vct{u}}{\vct{g}_i}
	> \sqrt{\frac{2}{\pi}} \cdot n - 2 \sqrt{n d} - t  \sqrt{n},
\end{equation}
except with probability $\econst^{-t^2/2}$.
\end{lemma}

\begin{proof} %
Add and subtract the mean from each summand on the left-hand side of~\eqref{eqn:gauss-lower-bd}
to obtain
\begin{equation} \label{eqn:mean-add-sub}
\inf_{\enorm{\vct{u}} = 1} \ \sum_{i=1}^n \absip{\vct{u}}{\vct{g}_i}
	\geq \inf_{\enorm{\vct{u}}=1} \ \sum_{i=1}^n \big[ \absip{\vct{u}}{\vct{g}_i} - \Expect \absip{\vct{u}}{\vct{g}_i} \big]
	+ \inf_{\enorm{\vct{u}}=1} \ \sum_{i=1}^n \Expect \absip{\vct{u}}{\vct{g}_i}
\end{equation}
The second sum on the right-hand side has a closed form expression because each term
is the expectation of a half-Gaussian random variable:
$\Expect \absip{\vct{u}}{\vct{g}_i} = \sqrt{2/\pi}$ for every unit vector $\vct{u}$.
Therefore,
\begin{equation}\label{eqn:half-gaussian}
	\inf_{\enorm{\vct{u}} = 1} \ \sum_{i=1}^{n}
    \Expect \absip{\vct{u}}{\vct{g}_i} = \sqrt{\frac{2}{\pi}} \cdot n.
\end{equation}
To control the first sum on the right-hand side of~\eqref{eqn:mean-add-sub},
we use a standard argument.  To bound the mean, we symmetrize the sum
and invoke a comparison theorem.  To control the probability of a large deviation,
we apply a measure concentration argument.

To proceed with the calculation of the mean, we use the Rademacher symmetrization lemma~\cite[Lem.~6.3]{Ledoux1991}
to obtain
\begin{equation*} \label{eqn:gauss-symmetrize}
\Expect \sup_{\enorm{\vct{u}} = 1} \ \sum_{i=1}^n \big[ \big(\Expect \absip{ \vct{u} }{ \vct{g}_i }\big)
	- \absip{ \vct{u} }{ \vct{g}_i } \big]
	\leq 2 \Expect \sup_{\enorm{\vct{u}} = 1} \ \sum_{i=1}^n \eps_i \absip{ \vct{u} }{ \vct{g}_i }.
\end{equation*}
The random variables $\eps_1, \dots, \eps_n$ are i.i.d.~Rademacher random variables that
are independent from the Gaussian sequence.
Next, invoke the Rademacher comparison theorem~\cite[Eqn.~(4.20)]{Ledoux1991} with
the function $\phi(\cdot) = \abs{\cdot}$ to obtain the further bound
$$
\Expect \sup_{\enorm{\vct{u}} = 1} \ \sum_{i=1}^n \big[ \big(\Expect \absip{ \vct{u} }{ \vct{g}_i }\big)
		- \absip{ \vct{u} }{ \vct{g}_i } \big]
	\leq 2 \Expect \sup_{\enorm{\vct{u}} = 1} \ \sum_{i=1}^n \eps_i \ip{ \vct{u} }{ \vct{g}_i }
	= 2 \Expect \enorm{ \sum\nolimits_{i=1}^n \eps_i \, \vct{g}_i }.
$$
The identity follows when we draw the sum into the inner product a maximize over all unit vectors.
From here, the rest of the argument is very easy.  Use Jensen's inequality to bound
the expectation by the root-mean-square, which has a closed form:
\begin{equation}\label{eqn:mean-estimate}
\Expect \sup_{\enorm{\vct{u}} = 1} \ \sum_{i=1}^n \big[ \big(\Expect \absip{ \vct{u} }{ \vct{g}_i }\big)
	- \absip{ \vct{u} }{ \vct{g}_i } \big]
	\leq 2 \left[ \Expect \enormsq{ \sum\nolimits_{i=1}^n \eps_i \, \vct{g}_i } \right]^{1/2}
	= 2 \sqrt{nd}.
\end{equation}
Note that the mean fluctuation~\eqref{eqn:mean-estimate}
is dominated by the centering term~\eqref{eqn:half-gaussian} when $n \gg d$.

To control the probability that the fluctuation term is large,
we use a standard concentration inequality~\cite[Theorem~1.7.6]{Bogachev1998}
for a Lipschitz function of independent Gaussian variables.  Define
a real-valued function on $d \times n$ matrices: \label{loc:f-defn}
$
f(\mtx{Z}) = \sup_{\enorm{\vct{u}} = 1} \ \sum_{i=1}^n ( \sqrt{2/\pi} - \absip{\vct{u}}{\vct{z}_i} ),
$
where $\vct{z}_i$ denotes the $i$th column of $\mtx{Z}$.  Compute that
\begin{equation*}
\abs{f(\mtx{Z}) - f(\mtx{Z}')}
	\leq \sup_{\enorm{\vct{u}} = 1} \ \sum_{i=1}^n \absip{\vct{u}}{\vct{z}_i - \vct{z}_i'}
	\leq \sum_{i=1}^n \norm{ \vct{z}_i - \vct{z}_i' }
	\leq \sqrt{n} \fnorm{ \mtx{Z} - \mtx{Z}' }.
\end{equation*}
Therefore, $f$ has Lipschitz constant $\sqrt{n}$ with respect to the Frobenius norm.
In view of the estimate~\eqref{eqn:mean-estimate} for the mean,
the Gaussian concentration bound implies that
\begin{equation} \label{eqn:bogachev-bd}
\Prob{ \sup_{\enorm{\vct{u}}=1} \ \sum_{i=1}^n \left[ \big(\Expect \absip{\vct{u}}{\vct{g}_i}\big)
	- \absip{\vct{u}}{\vct{g}_i} \right]
	\geq 2\sqrt{nd} + t \sqrt{n} } \leq \econst^{-t^2/2}.
\end{equation}
Introduce the bound~\eqref{eqn:bogachev-bd} and the identity~\eqref{eqn:half-gaussian} into~\eqref{eqn:mean-add-sub}
to complete the proof.
\end{proof}

\subsubsection{Tools for the Alignment Statistic}
\label{sec:prob-bounds-struct}

In this section, we develop the probability inequalities that
we need to estimate the alignment statistic $\Align(L)$ for
data drawn from the Haystack model.
First, we need a tail bound for the maximum singular value
of a Gaussian matrix.  The following inequality is a well-known consequence of Slepian's lemma.  See~\cite[Thm.~2.13]{Davidson2001} and the errata~\cite{Davidson2001-addenda}
for details.

\begin{prop}\label{prop:davidson-szarek}
Let $\mtx{G}$ be an $m \times n$ matrix whose entries are i.i.d.~standard normal
random variables.  For each $t \geq 0$,
\begin{equation*}
\Prob{ \norm{\mtx{G}} \geq \sqrt{m} + \sqrt{n} + t} < 1 - \Phi(t) < \econst^{-t^2/2},
\end{equation*}
where $\Phi(t)$ is the Gaussian cumulative density function
$$
\Phi(t) := \frac{1}{\sqrt{2\pi}} \int_{-\infty}^t \econst^{-\tau^2/2} \idiff{\tau}.
$$
\end{prop}

We also need a related result for random matrices with independent
columns that are uniformly distributed on the sphere.  The
argument bootstraps from Proposition~\ref{prop:davidson-szarek}.

\begin{lemma}\label{lem:sphere-sigma-upper-bound}
Let $\mtx{S}$ be an $m \times n$ matrix whose columns are
i.i.d.~random vectors distributed uniformly on the sphere
$\mathbb{S}^{m-1}$ in $\R^m$.  For each $t \geq 0$,
\begin{equation} \label{eqn:prob-sphere-sigma-deviation}
\Prob{\norm{\mtx{S}} \geq \frac{\sqrt{n}+\sqrt{m} + t}{\sqrt{m-0.5}}}
	\leq 1.5 \, \econst^{-t^2/2}.
\end{equation}
\end{lemma}

\begin{proof} %
Fix $\theta > 0$.  The Laplace transform method shows that
$$
P := \Prob{ \norm{\mtx{S}} \geq \frac{\sqrt{n} + \sqrt{m} + t}{\sqrt{m - 0.5}} }
	\leq \econst^{-\theta(\sqrt{n} + \sqrt{m} + t)} \cdot
	\Expect \econst^{\theta \sqrt{m - 0.5} \, \norm{\mtx{S}}}.
$$
We compare $\norm{\mtx{S}}$ with the norm of a Gaussian matrix by
introducing a diagonal matrix of $\chi$-distributed variables.
The rest of the argument is purely technical.

Let $\vct{r} = (r_1, \dots, r_n)$ be a vector of i.i.d.~$\chi$-distributed
random variables with $m$ degrees of freedom.  Recall that $r_i \, \tvct{g}_i \sim \vct{g}_i$,
where $\tvct{g}_i$ is uniform on the sphere and $\vct{g}_i$ is standard normal.
The mean of a $\chi$-distributed variable satisfies an inequality due to Kershaw~\cite{Kershaw1983}:
\begin{equation*} %
  \Expect r \geq\sqrt{m-0.5}
  \quad\text{when}\quad r \sim \chi_m.
\end{equation*}
Using Kershaw's bound and Jensen's inequality, we obtain
$$
\Expect \econst^{\theta \sqrt{m - 0.5} \, \norm{ \mtx{S} } }
	\leq \Expect \econst^{\theta \norm{ \Expect_{\vct{r}} \diag(\vct{r}) \mtx{S} }}
	\leq \Expect \econst^{\theta \norm{\mtx{G}}},
$$
where $\mtx{G}$ is an $m \times n$ matrix with i.i.d.~standard normal entries.

Define a random variable $Z := \norm{\mtx{G}} - \sqrt{n} - \sqrt{m}$, and let
$Z_+ := \max\{Z, 0\}$ denote its positive part.  Then
$$
\econst^{\theta t} \cdot P
	\leq \Expect \econst^{\theta Z}
	\leq \Expect \econst^{\theta Z_+}
	= 1 + \int_0^\infty \econst^{\theta \tau} \cdot \Prob{ Z_+ > \tau } \idiff{\tau}.
$$
Apply the cdf bound in Proposition~\ref{prop:davidson-szarek}, and identify the
complementary error function $\operatorname{erfc}$.
$$
\econst^{\theta t} \cdot P
	\leq 1 + \frac{\theta}{2} \int_0^\infty \econst^{\theta \tau}
		\cdot \operatorname{erfc}\left(\frac{\tau}{\sqrt{2}}\right) \idiff{\tau},
$$
A computer algebra system will report that this frightening integral has a closed form:
$$
\theta \int_0^\infty \econst^{\theta \tau}
	\cdot \operatorname{erfc}\left(\frac{\tau}{\sqrt{2}}\right) \idiff{\tau}
	= \econst^{\theta^2/2} \, ( \operatorname{erf}(\theta) + 1) - 1
	\leq 2 \, \econst^{\theta^2/2} - 1.
$$
We have used the simple bound $\operatorname{erf}(\theta) \leq 1$ for $\theta \geq 0$.
In summary,
$$
P \leq \econst^{-\theta t} \cdot \left[ \frac{1}{2} + \econst^{\theta^2/2} \right]
$$
Select $\theta = t$ to obtain the advertised bound~\eqref{eqn:prob-sphere-sigma-deviation}.
\end{proof}

\subsection{Proof of Theorem~\protect{\ref{thm:gaussian-outliers}}}
\label{sec:prf-gauss-outliers}

Suppose that the dataset $\coll{X}$ is drawn from the
Haystack model on page~\pageref{tab:haystack}.
Let $\mtx{X}_{\Out}$ be a $D \times N_{\Out}$ matrix
whose columns are the outliers $\vct{x} \in \Xout$, arranged in fixed order.
Recall that the inlier sampling ratio $\rho_{\In} := N_{\In} / d$
and the outlier sampling ratio $\rho_{\Out} := N_{\Out}/D$.

Let us begin with a lower bound for the permeance statistic
$\Perm(L)$.  The $N_{\In}$ inliers are drawn from a centered
Gaussian distribution on the $d$-dimensional
space $L$ with covariance $(\sigma_{\In}^2 / d) \, \Id_L$.
Rotational invariance and Lemma~\ref{lem:lower-bd-gauss},
with $t = c\sqrt{d}$, together imply that
the permeance statistic~\eqref{eqn:permeance} satisfies
$$
\Perm(L) > \frac{\sigma_{\In}}{\sqrt{d}}
	\left[ \sqrt{\frac{2}{\pi}} \cdot N_{\In} - (2 + c) \sqrt{N_{\In} d} \right]
	= \sigma_{\In} \sqrt{d} \left[ \sqrt{\frac{2}{\pi}} \rho_{\In} - (2+c) \sqrt{\rho_{\In}} \right],
$$
except with probability $\econst^{-c^2 d/2}$.

Next, we obtain an upper bound for the alignment statistic
$\Align(L)$.  The $N_{\Out}$ outliers are
independent, centered Gaussian vectors in~$\R^D$ with covariance
$(\sigma_\Out^2 / D) \, \Id$.  Proposition~\ref{prop:davidson-szarek},
with $t = c\sqrt{d}$ shows that
$$
\norm{\mtx{X}_{\Out}}
	\leq \frac{\sigma_{\Out}}{\sqrt{D}}
	\left[ \sqrt{N_{\Out}} + \sqrt{D} + c \sqrt{d} \right]
	= \sigma_{\Out} \left[ \sqrt{\rho_{\Out}} + 1 + c \sqrt{\frac{d}{D}} \right],
$$
except with probability $\econst^{-c^2 d/2}$.  Rotational invariance
implies that the columns of $\widetilde{ \Proj_{L^\perp} \mtx{X}_{\Out} }$
are independent vectors that are uniformly distributed
on the unit sphere of a $(D-d)$-dimensional space.
Lemma~\ref{lem:sphere-sigma-upper-bound} yields
$$
\vsmnorm{}{ \widetilde{ \Proj_{L^\perp} \mtx{X}_{\Out}} }
	\leq \frac{ \sqrt{N}_{\Out} + \sqrt{D - d} + c \sqrt{d} }{ \sqrt{D - d - 0.5} }
	< \sqrt{\frac{D}{D - d -0.5}} \left[ \sqrt{\rho_{\Out}} + 1 + c \sqrt{\frac{d}{D}} \right],
$$
except with probability $1.5 \, \econst^{-c^2 d / 2}$.  It follows that
$$
\Align(L) \leq \sigma_{\Out}\sqrt{\frac{D}{D - d -0.5}}
\left[ \sqrt{\rho_{\Out}} + 1 + c \sqrt{\frac{d}{D}} \right]^2
$$
except with probability $2.5 \, \econst^{-c^2 d/ 2}$.

Combining these bounds, we discover that the stability statistic satisfies
\begin{align*}
\Stab(L) &= \frac{\Perm(L)}{4\sqrt{d}} - \Align(L) \\
	&\geq \frac{\sigma_{\In}}{4} \left[ \sqrt{\frac{2}{\pi}} \rho_{\In} - (2+c) \sqrt{\rho_{\In}} \right] -
	\sigma_{\Out}\sqrt{\frac{D}{D - d -0.5}}
\left[ \sqrt{\rho_{\Out}} + 1 + c \sqrt{\frac{d}{D}} \right]^2
\end{align*}
except with probability $3.5 \, \econst^{-c^2 d / 2}$.
This completes the argument.

\section{Analysis of the IRLS Algorithm}
\label{sec:irls-analysis}

This appendix contains the details of our analysis of the IRLS
method, Algorithm~\ref{alg:IRLS}.  First, we verify that
Algorithm~\ref{alg:weighted-ls} reliably solves the
weighted least-squares subproblem~\eqref{eqn:irls-obj}.
Then, we argue that IRLS converges to a point near the
true optimum of the \rrp\ problem~\eqref{eqn:rrp}.

\subsection{Solving the Weighted Least-Squares Problem}
\label{sec:proof-of-subproblem-claim}

In this section, we verify that Algorithm~\ref{alg:weighted-ls}
correctly solves the weighted least-squares problem~\eqref{eqn:irls-obj}.
The following lemma provides a more mathematical statement of the
algorithm, along with the proof of correctness.  Note that this statement is slightly more general than the recipe presented in Algorithm~\ref{alg:weighted-ls} because it is valid for over the entire range \(0<d<D\).

\begin{lemma}[Solving the Weighted Least-Squares Problem]
\label{lem:subproblem-solved}
Assume that $0 < d < D$, and
suppose that $\Xall$ is a set of observations in $\R^D$.
For each $\vct{x} \in \Xall$, let $\beta_{\vct{x}}$ be a nonnegative
weight.
Form the weighted sample covariance matrix $\mtx{C}$, and
compute its eigenvalue decomposition:
$$
\mtx{C} := \sum_{\vct{x} \in \Xall} \beta_{\vct{x}} \,
\vct{xx}^\transp = \mtx{U \Lambda U}^\transp
\quad\text{where $\lambda_1 \geq \dots \geq \lambda_D \geq 0$.}
$$
When $\rank(\mtx{C}) \leq d$, construct a vector $\vct{\nu} \in \R^D$ via the formula
\begin{equation} \label{eqn:subprob-opt-degenerate}
\vct{\nu}
	:= ( \underbrace{1, \ \dots,\ 1}_{\text{$\lfloor d \rfloor$ times}},\
	d - \lfloor d \rfloor,\ 0,\  \dots,\ 0 )^\transp.
\end{equation}
When $\rank(\mtx{C}) > d$, define the positive quantity $\theta$ implicitly by
solving the equation
\begin{equation} \label{eqn:water-level}
\sum_{i=1}^D \frac{[\lambda_i - \theta]_+}{\lambda_i} = d.
\end{equation}
Construct a vector $\vct{\nu} \in \R^D$ whose components are
\begin{equation} \label{eqn:subprob-nondegenerate}
\nu_i := \frac{[\lambda_i - \theta]_+}{\lambda_i}
\quad\text{for $i = 1, \dots, D$.}
\end{equation}
In either case, an optimal solution to~\eqref{eqn:irls-obj} is given by
\begin{equation} \label{eqn:subprob-P-star}
\mtx{P}_{\star} := \mtx{U} \cdot \diag( \vct{\nu} ) \cdot \mtx{U}^\transp.
\end{equation}
In this statement, we enforce the
convention $0/0 := 0$, and $\diag$ forms a diagonal matrix from a
vector.
\end{lemma}

\begin{proof}[Proof of Lemma~\ref{lem:subproblem-solved}]
First, observe that the construction~\eqref{eqn:subprob-P-star} yields a matrix
$\mtx{P}_{\star}$ that satisfies the constraints of~\eqref{eqn:irls-obj}
in both cases.

When $\rank(\mtx{C}) \leq d$, we can verify that our construction of
the vector~$\vct{\nu}$ yields a optimizer of~\eqref{eqn:irls-obj}
by showing that the objective value is zero, which is minimal.
Evaluate the objective function~\eqref{eqn:irls-obj} at the point $\mtx{P}_{\star}$ to see that
\begin{equation} \label{eqn:irls-obj-nu}
\sum_{\vct{x} \in \Xall} \beta_{\vct x}\enormsq{ (\Id - \mtx{P}_{\star}) \, \vct{x} }
	= \trace \left[ (\Id - \mtx{P}_{\star}) \mtx{C}(\Id - \mtx{P}_{\star}) \right]
	= \sum_{i=1}^D (1 - \nu_i)^2 \, \lambda_i
\end{equation}
by definition of $\mtx{C}$ and the fact that $\mtx{C}$ and $\mtx{P}_{\star}$ are simultaneously diagonalizable.
The nonzero eigenvalues of $\mtx{C}$ appear among
$\lambda_1, \dots, \lambda_{\lfloor d \rfloor}$.
At the same time, $1 - \nu_i = 0$ for each $i = 1, \dots, \lfloor d \rfloor$.
Therefore, the value of~\eqref{eqn:irls-obj-nu} equals zero at $\mtx{P}_{\star}$.

Next, assume that $\rank(\mtx{C}) > d$.
The objective function in~\eqref{eqn:irls-obj} is convex, so we can verify that
$\mtx{P}_{\star}$ solves the optimization problem if the directional derivative of the objective
at $\mtx{P}_{\star}$ is nonnegative in every feasible direction.  A matrix $\mtx{\Delta}$
is a feasible perturbation if and only if
\begin{equation*} \label{eqn:weighted-ls-feas}
\mtx{0} \psdle \mtx{P}_{\star} + \mtx{\Delta} \psdle \Id
\quad\text{and}\quad
\trace \mtx{\Delta} = 0.
\end{equation*}
Let $\mtx{\Delta}$ be an arbitrary matrix that satisfies these constraints.
By expanding the objective of~\eqref{eqn:irls-obj} about $\mtx{P}_{\star}$,
easily compute the derivative in the direction $\mtx{\Delta}$.  In particular,
the condition
\begin{equation} \label{eqn:derivative-irls-obj}
- \ip{ \mtx{\Delta} }{ (\Id - \mtx{P}_{\star}) \mtx{C} } \geq 0
\end{equation}
ensures that the derivative increases in the direction $\mtx{\Delta}$. We now set about verifying~\eqref{eqn:derivative-irls-obj} for our choice of \(\mtx P_\star\) and all feasible \(\mtx \Delta\).

Note first that the quantity $\theta$ can be defined.  Indeed,
the left-hand side of~\eqref{eqn:water-level} equals $\rank(\mtx{C})$ when $\theta = 0$,
and it equals zero when $\theta \geq \lambda_1$.  By continuity, there exists a value of
$\theta$ that solves the equation.  Let $i_{\star}$ be the largest index where
$\lambda_{i_{\star}} > \theta$, so that $\nu_{i} = 0$ for each $i > i_{\star}$.
Next, define $M$ to be the subspace spanned by the eigenvectors
$\vct{u}_{i_{\star} + 1}, \dots, \vct{u}_D$.
Since $\nu_i$ is the eigenvalue of $\mtx{P}_{\star}$ with eigenvector $\vct{u}_i$,
we must have $\Proj_{M} \mtx{P}_{\star} \Proj_M = \mtx{0}$.
It follows that $\Proj_M \mtx{\Delta} \Proj_M \psdge \mtx{0}$
because $\Proj_M (\mtx{P}_{\star} + \mtx{\Delta}) \Proj_M \psdge \mtx{0}$.

To complete the argument, observe that
$$
(1 - \nu_i) \lambda_i = \lambda_i - (\lambda_i - \theta)_+ = \min\{ \lambda_i, \theta \}.
$$
Therefore,
$(\Id - \mtx{P}_{\star}) \mtx{C} = \mtx{U} \cdot \diag( \min\{ \lambda_i, \theta \} ) \cdot \mtx{U}^\transp$.
Using the fact that $\trace \mtx{\Delta} = 0$, we obtain
\begin{align*}
\ip{ \mtx{\Delta} }{ (\Id - \mtx{P}_{\star}) \mtx{C} }
	&= \ip{ \mtx{\Delta} }{ \mtx{U} \cdot \diag( \min\{ \lambda_i, \theta \} - \theta ) \cdot \mtx{U}^\transp } \\
	&= \smip{ \mtx{\Delta} }{ \underbrace{\mtx{U} \cdot
	\diag( 0, \dots, 0, \lambda_{i_{\star} + 1} - \theta, \dots, \lambda_D - \theta) \cdot
	\mtx{U}^{\transp} }_{=: \mtx{Z} }}
\end{align*}
Since $\lambda_i \leq \theta$ for each $i > i_{\star}$, each eigenvalue of $\mtx{Z}$ is nonpositive.
Furthermore, $\Proj_M \mtx{Z} \Proj_M = \mtx{Z}$.  We see that
$$
\ip{ \mtx{\Delta} }{ (\Id - \mtx{P}_{\star}) \mtx{C} }
	= \ip{ \mtx{\Delta}}{ \Proj_M \mtx{Z} \Proj_M }
	= \ip{ \Proj_M \mtx{\Delta} \Proj_M }{ \mtx{Z} }
	\leq 0,
$$
because the compression of $\mtx{\Delta}$ on $M$ is positive semidefinite and $\mtx{Z}$ is negative semidefinite.
In other words,~\eqref{eqn:derivative-irls-obj} is satisfied for every feasible perturbation $\mtx{\Delta}$
about $\mtx{P}_{\star}$.
\end{proof}

\subsection{Convergence of IRLS}
\label{sec:proof-of-lin-conv}

In this section, we argue that the IRLS method of Algorithm~\ref{alg:IRLS}
converges to a point whose value is nearly optimal for the \rrp\ problem~\eqref{eqn:rrp}.
The proof consists of two phases.  First, we explain how to modify the argument
from~\cite{CM99:Convergence-Lagged} to show that the iterates $\mtx
P^{(k)}$ converge to a matrix $\mtx P_\delta$, which is characterized as
the solution to a regularized counterpart of~\rrp. The fact
that the limit point $\mtx P_\delta$ achieves a near-optimal value for \rrp\
follows from the characterization.

\begin{proof}[Proof sketch for Theorem~\ref{thm:IRLS-conv}]
  We find it more convenient to work with the
  variables $\mtx Q := \Id -\mtx P$ and $\mtx{Q}^{(k)} :=
  \Id-\mtx{P}^{(k)}$.  First, let us define a regularized objective.
  For a parameter $\delta > 0$, consider the Huber-like function
  \begin{equation*} \renewcommand{\arraystretch}{1.5}
    H_\delta(x,y) = \left\{\begin{array}{ll}
        \frac{1}{2}\left(\frac{x^2}{\delta} + \delta\right), & 0 \leq y \leq \delta \\
        \frac{1}{2} \left(\frac{x^2}{y} + y\right), & y\geq \delta.
      \end{array}\right.
  \end{equation*}
  We introduce the convex function
\begin{align*}
F(\mtx{Q}) &:= \sum_{\vct{x} \in \Xall } H_\delta(\enorm{\mtx{Q}\vct{x}}, \ \enorm{\mtx{Q} \vct{x}})  \\
    &= \sumnl_{\{\vct x: \enorm{\mtx{Q}\vct{x}} \geq \delta \}} \enorm{\mtx{Q}\vct{x}}
    + \frac{1}{2} \sumnl_{\{\vct x : \enorm{\mtx{Q}\vct{x}} < \delta \}}
    \left(\frac{\enormsq{\mtx{Q} \vct{x}}}{\delta}+ \delta \right).
\end{align*}
  The second identity above highlights the interpretation of $F$ as
  a regularized objective function for~\eqref{eqn:rrp} under the assignment
  $\mtx{Q} = \Id -\mtx{P}$.  Note that $F$ is
  continuously differentiable at each matrix $\mtx{Q}$, and the gradient
\begin{equation*}
    \nabla F(\mtx{Q}) = \sum_{\vct x \in \Xall}
    	\frac{\mtx{Q} \vct{xx}^\transp}{\max\{\enorm{\mtx{Q} \vct{x}}, \ \delta\}}.
  \end{equation*}
  The technical assumption that the observations do not lie in the
  union of two strict subspaces of $\R^D$ implies that $F$ is \emph{strictly}
  convex; compare with the proof~\cite[Thm.~2]{ZL11:Novel-M-Estimator}.
  We define $\mtx{Q}_{\delta}$ to be the solution of a constrained optimization problem:
  \begin{equation*}
    \mtx{Q}_\delta := \argmin_{\substack{\mtx{0} \psdle \mtx{Q} \psdle \Id \\ \trace \mtx{Q} = D - d}}
    \ F(\mtx{Q}).
  \end{equation*}
The strict convexity of $F$ implies that $\mtx{Q}_{\delta}$
  is well defined.

  The key idea in the proof is to show
  that the iterates $\mtx{Q}^{(k)}$ of Algorithm~\ref{alg:IRLS}
  converge to the optimizer $\mtx{Q}_\delta$ of the regularized objective
  function~$F$.  We demonstrate that Algorithm~\ref{alg:IRLS} is a
  generalized Weiszfeld method in the sense of~\cite[Sec.~4]{CM99:Convergence-Lagged}.
  After defining some additional
  auxiliary functions and facts about these functions, we explain how
  the argument of~\cite[Lem.~5.1]{CM99:Convergence-Lagged} can be adapted to prove that
  the iterates of $\mtx{Q}^{(k)}=\Id -\mtx{P}^{(k)}\to \mtx{Q}_\delta$.
  The only innovation required is an inequality from convex analysis
  that lets us handle the constraints $\mtx{0} \psdle \mtx{Q} \psdle \Id$
and $\trace \mtx{Q} = D-d$.

  Now for the definitions.  We introduce the potential function
  \begin{equation*}
    G(\mtx{Q},\ \mtx{Q}^{(k)}) :=
    \sum_{\vct x \in \Xall} H_\delta(\enorm{\mtx{Q} \vct{x}}, \
    \vsmnorm{}{ \mtx{Q}^{(k)} \vct{x}} ).
  \end{equation*}
  Then $G(\cdot,\ \mtx Q^{(k)})$ is a smooth quadratic function. By
  collecting terms, we may relate $G$ and $F$ through the expansion
  \begin{equation*}
    G(\mtx Q,\ \mtx Q^{(k)}) = F(\mtx Q^{(k)}) + \ip{\mtx Q-\mtx Q^{(k)}}
      {\nabla F(\mtx Q^{(k)})} + \frac{1}{2}\ip{\mtx Q-\mtx Q^{(k)}}{C(\mtx
        Q^{(k)}) (\mtx Q - \mtx Q^{(k)})},
  \end{equation*}
  where $C$ is the continuous function
  \begin{equation*}
    C(\mtx Q^{(k)}) := \sum_{\vct{x} \in \Xall}\frac{\vct{xx}^\transp}{\max\{\enorm{\mtx{Q}\vct{x}},\ \delta\}}.
  \end{equation*}

Next, we verify some facts related to
  Hypothesis~4.2 and~4.3 of~\cite[Sec.~4]{CM99:Convergence-Lagged}.  Note that
  $F(\mtx{Q})= G(\mtx{Q},\ \mtx{Q})$.  Furthermore,
  $F(\mtx{Q}) \leq G(\mtx{Q},\ \mtx{Q}^{(k)})$ because $H_\delta(x, x) \leq H_\delta(x,y)$,
  which is a direct consequence of the AM--GM inequality.

  We now relate the iterates of Algorithm~\ref{alg:IRLS} to the
  definitions above. Given that $\mtx{Q}^{(k)} = \Id -\mtx{P}^{(k)}$,
  Step~2b of Algorithm~\ref{alg:IRLS}
  is equivalent to the iteration
  \begin{equation*}
    \bQ^{(k+1)}=\argmin_{\substack{\mtx{0} \psdle \mtx{Q} \psdle \Id \\ \trace \mtx{Q}  = D-d}}
    G(\mtx{Q},\ \mtx{Q}^{(k)}).
  \end{equation*}
  From this characterization, we have the monotonicity property
  \begin{equation}\label{eqn:monotone}
   F(\mtx{Q}^{(k+1)})
   \leq G(\mtx{Q}^{(k+1)}, \ \mtx{Q}^{(k)})
   \leq G(\mtx{Q}^{(k)},\ \mtx{Q}^{(k)})
   = F(\mtx{Q}^{(k)}).
 \end{equation}
 This fact motivates the stopping criterion for
 Algorithm~\ref{alg:IRLS} because it implies the objective values are
 decreasing: $ \alpha^{(k+1)} = F(\mtx{Q}^{k+1}) \leq F(\mtx{Q}^{(k)}) = \alpha^{(k)}$.

 We also require some information regarding the bilinear form induced
 by $C$.  Introduce the quantity $m := \max\left\{ \delta, \  \max_{\vct x \in \Xall}\{
 \enorm{\vct{x}}\}\right\}$.  Then, by symmetry of the matrix $\mtx{Q}$, and the
 fact that the inner product between positive semidefinite matrices is
 nonnegative we have
  \begin{equation*}
    \ip{\mtx{Q}}{C(\mtx{Q}^{(k)})\mtx{Q}}
    \geq \frac{1}{m} \trace\left(\mtx{Q}^2 \sum_{\vct{x} \in \Xall} \vct{xx}^\transp \right)
    \geq \fnormsq{\mtx{Q}}
    \underbrace{\left(\frac{ \lambda_{\min}\left(\sumnl_{\vct{x} \in \Xall} \vct{xx}^\transp \right)}
        {m} \right)}_{=: \, \mu}.
  \end{equation*}
  The technical assumption that the observations do not lie in two
  strict subspaces of $\R^D$ implies in particular that the
  observations span $\R^D$.  We deduce that $\mu>0$.

 Now we discuss the challenge imposed by the constraint set.  When the
 minimizer $\mtx{Q}^{(k+1)}$ lies on the boundary of the constraint
 set, the equality~\cite[Eqn.~(4.3)]{CM99:Convergence-Lagged} may not hold.
 However, if we denote the gradient of $G$ with respect to its first
 argument by $G_{\mtx{Q}}$, the \emph{in}equality
 \begin{align}
   \label{eqn:ineq-conv-4.3}
   0 &\leq \ip{\mtx{Q} - \mtx{Q}^{(k+1)}}{G_{\mtx{Q}}(\mtx{Q}^{(k+1)},\ \mtx{Q}^{(k)}) } \notag \\
   &= \ip{\mtx{Q} - \mtx{Q}^{(k+1)}}{\nabla F(\mtx{Q}^{(k)}) + C(\mtx{Q}^{(k)})(\mtx{Q}^{(k+1)}-\mtx{Q}^{(k)})}
 \end{align}
 holds for every $\mtx{Q}$ in the feasible set.  This is simply
 the first-order necessary and sufficient condition for the
 constrained minimum of a smooth convex function over a convex set.

 With the facts above, a proof that the iterates $\mtx{Q}^{(k)}$
 converge to $\mtx{Q}_\delta$ follows the
 argument of~\cite[Lem.~5.1]{CM99:Convergence-Lagged} nearly line-by-line.
 However, due to inequality~\eqref{eqn:ineq-conv-4.3}, the final conclusion
 is that, at the limit point $\overline{\mtx{Q}}$, the inequality
 $\ip{\mtx{Q} - \overline{\mtx{Q}}}{\nabla F(\overline{\mtx{Q}})}\geq 0$ holds for
 all feasible $\mtx{Q}$.  This inequality characterizes the
 global minimum of a convex function over a convex set, so the limit point
 must indeed be a global minimizer.  That is, $\overline{\mtx{Q}} = \mtx{Q}_{\delta}$.
 In particular, this argument shows that the iterates
 $\mtx{P}^{(k)}$ converge to $\mtx{P}_\delta := \Id - \mtx{Q}_{\delta}$
 as $k\to \infty$.

 The only remaining claim is that $\mtx{P}_\delta = \Id - \mtx{Q}_{\delta}$
 nearly minimizes~\eqref{eqn:rrp}.  We abbreviate
 the objective of~\eqref{eqn:rrp} under the identification $\mtx Q = \Id - \mtx{P}$
 by
 \begin{equation*}
   F_0(\mtx{Q}) := \sumnl_{x\in \Xall} \enorm{\mtx{Q} \vct{x} }.
 \end{equation*}
 Define $\mtx{Q}_{\star} := \argmin F_0(\mtx{Q})$ with respect to the
 feasible set $\mtx{0} \psdle \mtx{Q} \psdle \Id$ and $\trace(\mtx{Q}) = D - d$.
 From the easy inequalities $x \leq H_\delta(x,x) \leq x + \frac{1}{2} \delta$
 for $x \geq 0$, we see that
 \begin{equation*}
   0 \leq F(\mtx{Q}) - F_0(\mtx{Q})  \leq \frac{1}{2} \delta \abs{ \Xall }.
 \end{equation*}
 Evaluate the latter inequality at $\mtx Q_\delta$, and subtract the result from
 the inequality evaluated at $\mtx{Q}_\star$ to reach
 \begin{equation*}
   \bigl(F(\mtx{Q}_\star) - F(\mtx{Q}_\delta)\bigr) + \bigl(F_0(\mtx{Q}_\delta) - F_0(\mtx{Q}_{\star}) \bigr)
	\leq \frac{1}{2} \delta \abs{\Xall}.
 \end{equation*}
 Since $\mtx{Q}_\delta$ and $\mtx{Q}_\star$ are optimal for their respective problems,
 both terms in parenthesis above are positive, and we deduce that
 $F_0(\mtx{Q}_\delta) - F_0(\mtx{Q}_\star) \leq \frac{1}{2} \delta \abs{\Xall}$.  Since
 $F_0$ is the objective function for~\eqref{eqn:rrp} under the map
 $\mtx{P} = \Id - \mtx{Q}$, the proof is complete.
\end{proof}

\begin{acknowledgements}%
  Lerman and Zhang were supported in part by the IMA and by NSF grants
  DMS-09-15064 and DMS-09-56072.  McCoy and Tropp were supported by
  ONR awards N00014-08-1-0883 and N00014-11-1002, AFOSR award
  FA9550-09-1-0643, DARPA award N66001-08-1-2065, and a Sloan Research
  Fellowship.  The authors thank Eran Halperin, Yi Ma, Ben Recht,
  Amit Singer, and John Wright for helpful conversations.  The anonymous
  referees provided many thoughtful and incisive remarks that helped us
  improve the manuscript immensely.

\end{acknowledgements}
\bibliographystyle{spmpsci}
\bibliography{bib-rrp}

\end{document}

%% file: rrp-macro-file.tex
\usepackage{amsthm}
\usepackage{amsmath}
\usepackage{amssymb}
\usepackage{bm}
\usepackage{mathrsfs}
\usepackage{yhmath}

\usepackage{graphicx}
\usepackage[table]{xcolor}

\usepackage{algorithmic}
\usepackage{algorithm}

\usepackage{color}

\usepackage{url}
\usepackage{microtype}

\usepackage{supertabular}
\usepackage{enumitem}

\definecolor{litegray}{rgb}{0.95,0.95,0.95}
\definecolor{liteblue}{rgb}{0.85,0.95,1}
\definecolor{liteyellow}{rgb}{1,1,.8}

\newtheorem{thm}{Theorem}

\newtheorem{lemma}[thm]{Lemma}
\newtheorem{sublemma}[thm]{Sublemma}
\newtheorem{prop}[thm]{Proposition}

\theoremstyle{definition}

\newtheorem{fact}[thm]{Fact}

\theoremstyle{remark}

\numberwithin{equation}{section}

\numberwithin{thm}{section}

\numberwithin{figure}{section}
\numberwithin{table}{section}
\numberwithin{algorithm}{section}

\DeclareMathAlphabet{\mathsfsl}{OT1}{cmss}{m}{sl}

\newcommand{\lang}{\textit}

\newcommand{\term}{\emph}

\renewcommand{\phi}{\varphi}

\newcommand{\eps}{\varepsilon}

\newcommand{\sumnl}{\sum\nolimits}
\newcommand{\econst}{\mathrm{e}}

\newcommand{\zerovct}{\vct{0}}

\newcommand{\Id}{\mathbf{I}}

\newcommand{\coll}[1]{\mathscr{#1}}

\newcommand{\R}{\mathbb{R}}

\newcommand{\abs}[1]{\left\vert {#1} \right\vert}
\newcommand{\abssq}[1]{{\abs{#1}}^2}

\newcommand{\diff}[1]{\mathrm{d}{#1}}
\newcommand{\idiff}[1]{\, \diff{#1}}

\newcommand{\argmin}{\operatorname*{arg\; min}}

\newcommand{\Prob}[1]{\mathbb{P}\left\{ {#1} \right\}}
\newcommand{\Expect}{\operatorname{\mathbb{E}}}

\newcommand{\normal}{\textsc{normal}}

\newcommand{\vct}[1]{\bm{#1}}
\newcommand{\mtx}[1]{\bm{#1}}

\newcommand{\tvct}[1]{\widetilde{\vct{#1}}} % Add tilde to vector

\newcommand{\transp}{{\sf t}}

\newcommand{\range}{\operatorname{range}}

\newcommand{\rank}{\operatorname{rank}}

\newcommand{\diag}{\operatorname{diag}}
\newcommand{\trace}{\operatorname{tr}}

\newcommand{\Proj}{\ensuremath{\mtx{\Pi}}}

\newcommand{\psdle}{\preccurlyeq}
\newcommand{\psdge}{\succcurlyeq}

\newcommand{\ip}[2]{\left\langle {#1},\ {#2} \right\rangle}
\newcommand{\absip}[2]{\abs{\ip{#1}{#2}}}
\newcommand{\abssqip}[2]{\abssq{\ip{#1}{#2}}}

\newcommand{\norm}[1]{\left\Vert {#1} \right\Vert}

\newcommand{\smip}[2]{\bigl\langle {#1}, \ {#2} \bigr\rangle}

\newcommand{\vsmnorm}[2]{{\Vert {#2} \Vert}_{#1}}

\newcommand{\enorm}[1]{\norm{#1}}
\newcommand{\enormsq}[1]{\enorm{#1}^2}

\newcommand{\fnorm}[1]{\norm{#1}_{\mathrm{F}}}
\newcommand{\fnormsq}[1]{\fnorm{#1}^2}

\newcommand{\pnorm}[2]{\norm{#2}_{#1}}

\newcommand{\minimize}{\textsf{minimize}\quad}
\newcommand{\subjto}{\quad\textsf{subject to}\quad} 

\newcommand{\bigOh}{\mathcal{O}}

\newcommand{\In}{\ensuremath{\mathrm{in}}}
\newcommand{\Out}{\ensuremath{\mathrm{out}}}
\newcommand{\Noise}{\ensuremath{\mathrm{noise}}}
\newcommand{\SNR}{\ensuremath{\mathrm{SNR}}}
\newcommand{\Xin}{\ensuremath{\mathcal{X}_{\mathrm{in}}}}
\newcommand{\Xout}{\ensuremath{\mathcal{X}_{\mathrm{out}}}}
\newcommand{\Xall}{\ensuremath{\mathcal{X}}}

\newcommand{\Perm}{\ensuremath{\mathcal{P}}}
\newcommand{\Resid}{\ensuremath{\mathcal{R}}}
\newcommand{\Align}{\ensuremath{\mathcal{A}}}

\newcommand{\Stab}{\ensuremath{\mathcal{S}}}

\newcommand{\bQ}{\boldsymbol{Q}}

\newcommand{\rrp}{{\sc reaper}}
\newcommand{\srrp}{{\sc s-reaper}}